\documentclass[10pt,journal,epsfig]{IEEEtran}

\usepackage[dvips]{graphicx}
\usepackage{graphicx}
\usepackage{amssymb}
\usepackage{cite}
\usepackage{subfig}
\usepackage{amsmath}
\usepackage{algorithm}
\usepackage{algorithmic}
\usepackage{multirow}
\usepackage{url}

\begin{document}

\title{Fast Compressed Power Spectrum Estimation: Towards A Practical Solution for Wideband Spectrum Sensing}

\author{Linxiao Yang, Jun Fang, ~\IEEEmembership{Senior Member},
Huiping Duan and Hongbin Li, ~\IEEEmembership{Fellow,~IEEE}
\thanks{Linxiao Yang and Jun Fang are with the National Key Laboratory
of Science and Technology on Communications, University of
Electronic Science and Technology of China, Chengdu 611731, China,
Email: JunFang@uestc.edu.cn}
\thanks{Huiping Duan is with the School of Information and Communication Engineering,
University of Electronic Science and Technology of China, Chengdu
611731, China, Email: huipingduan@uestc.edu.cn}
\thanks{Hongbin Li is
with the Department of Electrical and Computer Engineering,
Stevens Institute of Technology, Hoboken, NJ 07030, USA, E-mail:
Hongbin.Li@stevens.edu}
\thanks{This work was supported in part by the National Science
Foundation of China under Grant 61871091.}
\thanks{\copyright 2019 IEEE. Personal use of this material is permitted.
	Permission from IEEE must be obtained for all other uses, in any current or future media, including reprinting/republishing this material for advertising or promotional purposes, creating new collective works, for resale or redistribution to servers or lists, or reuse of any copyrighted component of this work in other works.}}

\maketitle

\begin{abstract}
There has been a growing interest in wideband spectrum sensing due
to its applications in cognitive radios and electronic
surveillance. To overcome the sampling rate bottleneck for
wideband spectrum sensing, in this paper, we study the problem of
compressed power spectrum estimation whose objective is to
reconstruct the power spectrum of a wide-sense stationary signal
based on sub-Nyquist samples. By exploring the sampling structure
inherent in the multicoset sampling scheme, we develop a
computationally efficient method for power spectrum
reconstruction. An important advantage of our proposed method over
existing compressed power spectrum estimation methods is that our
proposed method, whose primary computational task consists of fast
Fourier transform (FFT), has a very low computational complexity.
Such a merit makes it possible to efficiently implement the
proposed algorithm in a practical field-programmable gate array
(FPGA)-based system for real-time wideband spectrum sensing. Our
proposed method also provides a new perspective on the power
spectrum recovery condition, which leads to a result similar to
what was reported in prior works. Simulation results are presented
to show the computational efficiency and the effectiveness of the
proposed method.
\end{abstract}

\begin{keywords}
Compressed sampling, wideband spectrum sensing, power spectrum
estimation.
\end{keywords}

\section{Introduction}
The scarcity of spectrum resource is a critical problem for future
wireless communications and Internet of Things
\cite{PalattellaDohler16}. Dynamic spectrum access
\cite{SongXin12} allows the idle spectrum to be used by unlicensed
users and provides a promising way to enhance the spectrum
efficiency. Real-time spectrum sensing, which empowers the
secondary user to identify spectrum holes, is a fundamental
technique to dynamic spectrum sensing and has received much
interest over the past decade \cite{YucekArslan09,AxellLeus12}.
However, most previous studies, e.g.
\cite{YucekArslan09,AxellLeus12,WangFang10}, focus on narrowband
spectrum sensing. Only a few touched the topic of wideband
spectrum sensing until recently, e.g. \cite{SunNallanathan13}. As
the cognitive radio networks are expected to operate over a wide
frequency range, wideband spectrum sensing is important for future
cognitive radio systems. In addition, in some applications such as
electronic surveillance, one also needs to perform wideband
spectrum sensing to identify the frequency locations of a number
of signals that spread over a wide frequency band. A conventional
receiver requires to sample the received signal at the Nyquist
rate, which may be too power hungry due to the use of high speed
analog-digital converters (ADCs) or even infeasible if the
spectrum under monitoring is very wide, say, of several GHz. One
way to overcome this difficulty is to divide the frequency
spectrum under monitoring into a number of separate frequency
segments and then sequentially scan these frequency channels.
Nevertheless, such a scanning scheme incurs a sensing latency and
may fail to capture short-lived signals \cite{WangFang18}.

The recently emerged technique compressed sensing \cite{Donoho06}
provides a new approach to sensing a wide frequency band via low
sampling rate ADCs. Motivated by the fact that the spectrum in
general is severely underutilized in a specific geographic area,
the rationale behind such schemes is to exploit the inherent
sparsity in the frequency domain and formulate wideband spectrum
sensing as a sparse signal recovery problem which, according to
the compressed sensing theory \cite{CandesRomberg06,Donoho06}, can
perfectly recover the signal of the entire frequency band based on
sub-Nyquist samples. Different compressed sampling architectures
and sparse signal recovery algorithms, e.g.
\cite{MishaliEldar10,ZhangHan11,WakinBecker12,MishaliEldar11,SunChiu12,QinGao16,KhalfiHamdaoui18},
were developed for wideband spectrum sensing over the past few
years. Such an approach, however, suffers several drawbacks.
Firstly, reconstructing the original signal of the entire
frequency band via sparse signal recovery methods involves a
prohibitively high computational complexity for practical systems.
Secondly, due to the noise across a wide frequency band, the
receiver usually has a low signal-to-noise ratio. While in the low
signal-to-noise ratio regime, sparse signal recovery methods yield
barely satisfactory recovery performance.

To overcome these difficulties, some works proposed to reconstruct
the power spectrum of the wideband signal, instead of the original
signal itself, from sub-Nyquist samples
\cite{TianTafesse12,ArianandaLeus12,YenTsai13,CohenEldar14,LagunasNajar15,RomeroAriananda16}.
Compared with recovering the original signal, reconstructing the
power spectrum can significantly reduce the amount of data being
stored and transmitted. Also, the use of statistical information
enables the proposed algorithm to accurately perform spectrum
sensing even in low signal-to-noise ratio environments. In
addition, it was shown that it is even possible to perfectly
reconstruct the power spectrum without placing any sparse
constraint on the wideband spectrum under monitoring
\cite{ArianandaLeus12,YenTsai13,CohenEldar14}. Overall, this
approach seems to be a more practical solution for wideband
spectrum sensing. Nevertheless, existing compressed power spectrum
estimation methods still incur a high computational complexity
that is impractical for real-time wideband spectrum sensing. To
address this difficulty, notice that the data samples obtained via
the multicoset sampling scheme are a subset of the Nyquist
samples. Such a property allows us to establish an amiable
relationship between the autocorrelation sequence and the
sub-Nyquist samples, based on which we develop a computationally
efficient algorithm for compressed power spectrum estimation. The
proposed algorithm has a low computational complexity which scales
linearly with the number of samples (in time) $L$ and the
downsampling factor $N$. In contrast, existing compressed power
spectrum estimation methods have a complexity either scaling
polynomially with $L$ or scaling quadratically with $N$. Also, our
proposed method, which involves the discrete Fourier transform
(DFT) as its computationally dominant task, can be efficiently
implemented via FFT algorithms. These advantages make it a
practical solution for real-time wideband spectrum sensing.

The rest of the paper is organized as follows. In Section
\ref{sec:review}, we provide a brief review of existing compressed
wideband power spectrum estimation methods. In Section
\ref{sec:proposed-method}, a computationally-efficient compressed
power spectrum estimation algorithm is proposed, along with an
analysis of the recovery condition and its computational
complexity. Simulation results are presented in Section
\ref{sec:simulation-results}, followed by concluding remarks in
Section \ref{sec:conclusions}.

\section{Review of State-of-The-Arts} \label{sec:review}
In this section, we provide a brief review of existing compressive
wideband power spectrum estimation methods. Generally, the
methods, according to the way of formulation, can be divided into
a time-domain approach and a frequency-domain approach.

\begin{figure}[!t]
\centering
\includegraphics[width=8cm]{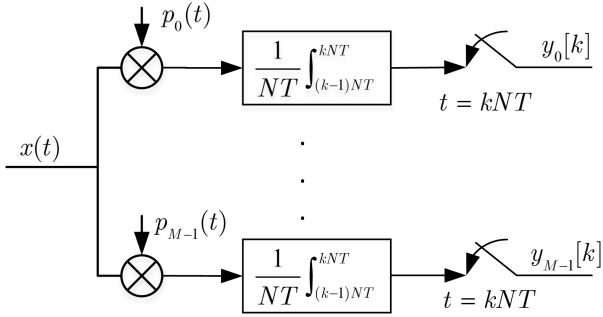}
\caption{Compressed sampling architecture: analog to information
converter (AIC)} \label{fig1}
\end{figure}

\subsection{Time-Domain Power Spectrum Reconstruction Approach}
The time-domain approach reconstructs the power spectrum of a
wide-sense stationary multi-band signal $x(t)$ by exploring a
relationship between the autocorrelation of the Nyquist samples
and the autocorrelation of the sub-Nyquist samples. Here the
sub-Nyquist samples are obtained via a compressed sampling scheme
such as the analog to information converter (AIC)
\cite{LaskaKirolos07} or the multicoset sampling scheme
\cite{VenkataramaniBresler01}. Let us take the AIC for an example.
The AIC (see Fig. \ref{fig1}) consists of a number of sampling
channels, also referred to as branches. In each branch, say branch
$m$, the signal $x(t)$ is modulated by a periodic pseudo-noise
(PN) sequence $p_{m}(t)$ of period $NT$, and then passes through
an integrate-and-dump device with period $NT$. Here $T$ denotes
the Nyquist sampling interval. $N$, a positive integer, is the
downsampling factor of user's choice, i.e. the sampling rate at
each branch is equal to $1/N$ times the Nyquist rate. Assume that
the PN sequence is a piecewise constant signal with constant
values in every interval of length $T$, i.e.
\begin{align}
p_m(t)=c_m[n] \quad nT\leq t<(n+1)T.
\end{align}
The output of the $m$th branch at the $l$th sampling instant can
be expressed as
\begin{align}
y_m[l]=\boldsymbol{c}^T_m\boldsymbol{x}[l],
\end{align}
where
\begin{align}
\boldsymbol{c}_m\triangleq
[c_m[0]\phantom{0}c_m[1]\phantom{0}\ldots\phantom{0}c_m[N-1]]^T,
\end{align}
and
\begin{align}
\boldsymbol{x}[l]\triangleq
[x[lN]\phantom{0}x[lN+1]\phantom{0}\ldots\phantom{0}x[(l+1)N-1]]^T,
\end{align}
in which $x[lN+n]$ is the average level of $x(t)$ at time interval
$[lNT+nT,lNT+(n+1)T]$.

Let $\boldsymbol{y}[l]\triangleq
[y_1[l]\phantom{0}\dots\phantom{0}y_M[l]]^T$ denote the data
sample vector collected at the $l$th sampling time instant. We
have
\begin{align}
\boldsymbol{y}[l]=\boldsymbol{C}\boldsymbol{x}[l],
\label{AIC-input-output-relation}
\end{align}
where $\boldsymbol{C}\triangleq
[\boldsymbol{c}_1\phantom{0}\ldots\phantom{0}\boldsymbol{c}_M]^T$.
Calculating the autocorrelation of $\boldsymbol{y}[l]$ yields:
\begin{align}
\boldsymbol{R}_{y}\triangleq&\mathbb{E}[\boldsymbol{y}[l]\boldsymbol{y}^H[l]]
=\boldsymbol{C}\mathbb{E}[\boldsymbol{x}[l]\boldsymbol{x}^H[l]]\boldsymbol{C}^H\nonumber\\
=&\boldsymbol{C}\boldsymbol{R}_{x}\boldsymbol{C}^H, \label{eqn1}
\end{align}
where
$\boldsymbol{R}_{x}\triangleq\mathbb{E}[\boldsymbol{x}[l]\boldsymbol{x}^H[l]]$.
Note that here $\boldsymbol{x}[l]$ can be considered as a vector
constructed by $N$ consecutive Nyquist samples. Hence (\ref{eqn1})
establishes the relationship between the autocorrelation of the
Nyquist samples and the autocorrelation of the sub-Nyquist
samples. We see that $\boldsymbol{R}_{x}$ is a Toeplitz matrix,
with its $(n+m,n)$th entry equal to $r_x[m]\triangleq
E[x[n+m]x^{\ast}[n]]$. Moreover, as $x(t)$ is a real-valued
signal, we have $r_x[m]=r_x[-m]$. Define
\begin{align}
\boldsymbol{r}_x\triangleq [r_x[0]\phantom{0}
r_x[1]\phantom{0}\dots \phantom{0} r_x[N-1]]^T.
\end{align}
Clearly, we can express the vectorized form of
$\boldsymbol{R}_{x}$ as:
$\text{vec}(\boldsymbol{R}_{x})=\boldsymbol{B}\boldsymbol{r}_x$,
where $\text{vec}(\cdot)$ denotes the vectorization operation,
$\boldsymbol{B}\in\{0,1\}^{N^2\times N}$ is the corresponding
selection matrix. Taking the vectorization of
$\boldsymbol{R}_{y}$, we obtain
\begin{align}
\text{vec}(\boldsymbol{R}_{y})=(\boldsymbol{C}^{*}\otimes\boldsymbol{C})\boldsymbol{B}\boldsymbol{r}_x
\triangleq \boldsymbol{\Phi}\boldsymbol{r}_x, \label{eqn-2}
\end{align}
where $\otimes$ denotes the Kronecker product, and
$\boldsymbol{\Phi}\triangleq
(\boldsymbol{C}^{*}\otimes\boldsymbol{C})\boldsymbol{B}$ is an
$M^2\times N$ matrix. Since $\boldsymbol{B}$ is full column rank,
we can let $M^2>N$ and carefully select the PN sequences such that
$\boldsymbol{\Phi}$ is a full column rank matrix. Thus
$\boldsymbol{r}_x$ can be easily calculated as
\begin{align}
\boldsymbol{r}_x=\boldsymbol{\Phi}^{\dag}\text{vec}(\boldsymbol{R}_{y}),
\end{align}
and the power spectrum of the original signal $x(t)$ can be
reconstructed by taking the discrete Fourier transform (DFT) of
$\boldsymbol{r}_x$.

It should be noted that the resolution of the recovered power
spectrum is determined by the length of the autocorrelation
sequence $\boldsymbol{r}_x$, i.e. $N$. On the other hand, we
require $M^2>N$ to make $\boldsymbol{\Phi}$ invertible and $M$,
the number of sampling branches, cannot be arbitrarily large due
to hardware complexity and cost. As a result, $N$ cannot be chosen
to be very large, which prevents us from obtaining a
high-resolution power spectrum. To address this issue, one can
collect multiple samples $\boldsymbol{y}[l],l=0,\ldots, L-1$ and
stack them into a long vector $\boldsymbol{y}$. Specifically,
define
\begin{align}
\boldsymbol{x}\triangleq &
[\boldsymbol{x}^T[0]\phantom{0}\dots\phantom{0}\boldsymbol{x}^T[L-1]]^T, \nonumber\\
\boldsymbol{y}\triangleq &
[\boldsymbol{y}^T[0]\phantom{0}\dots\phantom{0}\boldsymbol{y}^T[L-1]]^T.
\nonumber
\end{align}
Then we have
\begin{align}
\boldsymbol{y}=(\boldsymbol{I}_L\otimes\boldsymbol{C})\boldsymbol{x}\triangleq
\boldsymbol{\bar{C}}\boldsymbol{x},
\end{align}
where $\boldsymbol{I}_L$ denotes an $L\times L$ identity matrix.
Calculating the autocorrelation of $\boldsymbol{y}$ yields:
\begin{align}
\boldsymbol{\bar{R}}_y=\boldsymbol{\bar{C}}\boldsymbol{\bar{R}}_x\boldsymbol{\bar{C}}^H.
\end{align}
Similarly, the autocorrelation sequence
$\boldsymbol{\bar{r}}_x\triangleq
[r_x[0]\phantom{0}\dots\phantom{0}r_x[LN-1]]^T$ can be estimated
as
\begin{align}
\boldsymbol{\bar{r}}_x=((\boldsymbol{\bar{C}}^{*}\otimes\boldsymbol{\bar{C}})\boldsymbol{\bar{B}})^{\dag}
\text{vec}(\boldsymbol{\bar{R}}_y), \label{eqn5}
\end{align}
where $\boldsymbol{\bar{B}}\in\{0,1\}^{(LN)^2\times LN}$ is the
corresponding selection matrix. Note that the autocorrelation
sequence $\boldsymbol{\bar{r}}_x$ has a dimension of $NL$, from
which we can obtain a power spectrum of length $(2NL-1)$ with a
spectrum resolution $1/((2NL-1)T)$. Thus a desired high-resolution
power spectrum can be achieved by choosing a proper value of $L$.

In summary, the time-domain approach consists of three steps,
namely, calculating the correlation matrix
$\boldsymbol{\bar{R}}_y$, estimating $\boldsymbol{\bar{r}}_x$ via
(\ref{eqn5}), and reconstructing the power spectrum by taking the
Fourier transform of $\boldsymbol{\bar{r}}_x$. The correlation
matrix $\boldsymbol{\bar{R}}_y$ can be estimated as
\begin{align}
\boldsymbol{\hat{\bar{R}}}_y=\frac{1}{P}\sum_{p=1}^{P}\boldsymbol{y}_p\boldsymbol{y}^H_p,
\end{align}
where
\begin{align}
\boldsymbol{y}_p\triangleq &
[\boldsymbol{y}^T[p]\phantom{0}\dots\phantom{0}\boldsymbol{y}^T[L+p-1]]^T.
\end{align}
The matrix
$((\boldsymbol{\bar{C}}^{*}\otimes\boldsymbol{\bar{C}})\boldsymbol{\bar{B}})^{\dag}$
in (\ref{eqn5}) can be computed in advance. Thus we only need to
perform a matrix-vector product to obtain
$\bar{\boldsymbol{r}}_x$. We can easily verify that these three
steps involve $PL^2 M^2$, $L^3 M^2N$, and $(2LN-1)\log(2LN-1)$
floating-point operations, respectively, which scale polynomially
with the number of samples, $L$, making the time-domain approach
infeasible for practical systems. To see this, suppose we would
like to sense a wide frequency band up to 1GHz. Set $M=8$ and
$N=25$. In order to obtain a spectrum resolution of 10kHz, we need
to set $L=4000$, in which case the total number of floating-point
operations is at least as large as $2.5\times10^{13}$, and it will
take a few hours to complete the computational task even with a
high-performance FPGA. We note that a more efficient time-domain
approach was developed in \cite{ArianandaLeus12}, which has a
computational complexity similar to the frequency-domain approach
analyzed below.

\begin{figure}
    \centering
    \includegraphics[width=8cm]{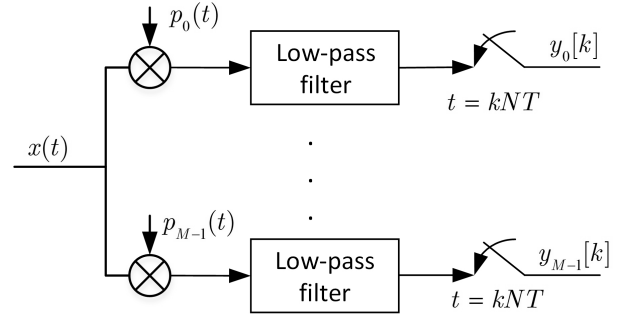}
    \caption{Compressed sampling architecture: modulated wideband converter (MWC)}
    \label{fig2}
\end{figure}

\subsection{Frequency-Domain Power Spectrum Reconstruction Approach}
In addition to the time-domain approach, another approach deals
with the power spectrum estimation problem from a frequency
viewpoint. The frequency-domain approach was originally proposed
in \cite{CohenEldar14}, in which the sub-Nyquist data samples are
obtained via a compressed sampling scheme termed as modulated
wideband converter (MWC) \cite{MishaliEldar09}. The MWC has an
architecture similar to the AIC, but replaces the integrate
devices with low-pass filters. Different from the time-domain
approach, the frequency-domain approach aims to establish the
relationship between the frequency representation of the original
signal and the frequency representation of the compressed samples.

Note that the PN sequence is periodic with a period $NT$. It,
therefore, has a discrete spectrum consisting of an infinite
number of impulses. According to the convolution properties, the
spectrum of the PN sequence-modulated signal in each branch is a
weighted combination of the spectrum of the original signal $x(t)$
and its frequency-shifted versions. By setting the cutoff
frequency of the low-pass filters to $1/(2NT)$, the spectrum of
the low-pass filtered signal at each branch is a superposition of
$N$ segments. Here the $N$ segments are obtained via dividing the
spectrum of the original signal $x(t)$ into $N$ segments of equal
width. Mathematically, the relationship between the spectrum of
the low-pass filtered signal and that of the original signal is
given as \cite{CohenEldar14}
\begin{align}
\boldsymbol{y}(f)=\boldsymbol{A}\boldsymbol{x}(f),
\label{yf-definition}
\end{align}
where the $m$th entry of $\boldsymbol{y}(f)$ denotes the Fourier
transform of the filtered signal at the $m$th branch,
$\boldsymbol{A}$ denotes the mixture matrix determined by the
spectrum of the PN sequences, and the $n$th element of
$\boldsymbol{x}(f)$, denoted by $x_n(f)$, is the $n$th segment of
the spectrum of the original signal, which is given as
\begin{align}
x_n(f)=
\begin{cases}
 X(f+\frac{2n-N-1}{2NT})& \text{if $N$ is odd}\\
 X(f+\frac{2n-N-2}{2NT})& \text{if $N$ is even}
\end{cases},
\end{align}
$\forall f\in (-\frac{1}{2NT},\frac{1}{2NT}]$, where $X(f)$
denotes the frequency spectrum of the original signal $x(t)$.

Calculating the correlation of $\boldsymbol{y}(f)$ gives
\begin{align}
\boldsymbol{R}_y(f)=&\mathbb{E}[\boldsymbol{y}(f)\boldsymbol{y}^H(f)]
=\boldsymbol{A}\mathbb{E}[\boldsymbol{x}(f)\boldsymbol{x}^H(f)]\boldsymbol{A}^H\nonumber\\
\triangleq&\boldsymbol{A}\boldsymbol{R}_x(f)\boldsymbol{A}^H.
\label{eqn3}
\end{align}
For a wide-sense stationary signal $x(t)$, its Fourier transform
$X(f)$ has the following property
\begin{align}
\mathbb{E}[X(f_1)X^{*}(f_2)]=
\begin{cases}
0,&f_1\ne f_2\\
P(f_1),&f_1=f_2
\end{cases},
\end{align}
where $P(f)$ is the power spectrum of $x(t)$. Thus
$\boldsymbol{R}_x(f)$ is a diagonal matrix and (\ref{eqn3}) can be
rewritten as
\begin{align}
\text{vec}(\boldsymbol{R}_y(f))=(\boldsymbol{A}^{*}\odot\boldsymbol{A})\boldsymbol{r}_{x}(f),
\label{eqn4}
\end{align}
where $\odot$ denotes the Khatri-Rao product, and
$\boldsymbol{r}_{x}(f)$ is a vector containing the diagonal
entries of $\boldsymbol{R}_x(f)$. It is easy to see that if
$\boldsymbol{A}^{*}\odot\boldsymbol{A}$ has a full column rank,
then $\boldsymbol{r}_{x}(f)$ can be estimated from
$\boldsymbol{R}_y(f)$, i.e.
\begin{align}
\boldsymbol{r}_{x}(f)=(\boldsymbol{A}^{*}\odot\boldsymbol{A})^{\dag}\text{vec}(\boldsymbol{R}_y(f)).
\label{eqn-6}
\end{align}
One can carefully design the PN sequences such that if $M^2\ge N$,
the matrix $\boldsymbol{A}^{*}\odot\boldsymbol{A}$ has a full
column rank. Note that the $n$th entry of $\boldsymbol{r}_{x}(f)$
is given as $r_{x,n}(f)=\mathbb{E}[x_n(f)x^{*}_n(f)]$. We
uniformly discretize the frequency range $(-1/(2NT),1/(2NT)]$ into
$2L$ grid points, say $f_1,f_2,\ldots,f_{2L}$ and estimate
$\boldsymbol{r}_{x}(f_l),\forall l\in\{1,\ldots,2L\}$ from
(\ref{eqn4}). Accordingly, the power spectrum of the original
signal can be reconstructed based on
$\{\boldsymbol{r}_{x}(f_l)\}$. Specifically, the reconstructed
power spectrum vector has a dimension of $2NL$, with its
$((n-1)L+1)$th entry equal to the $n$th element of
$\boldsymbol{r}_{x}(f_l)$. Clearly, the recovered power spectrum
has a resolution of $1/(2NLT)$. More details about the
frequency-domain approach can be found in \cite{CohenEldar14}.

We see that the frequency-domain method involves computing
$\boldsymbol{y}(f)$, $\boldsymbol{R}_y(f)$, and (\ref{eqn-6}). To
compute $\boldsymbol{y}(f)$, we need to perform $2L$-point DFT of
the time-domain signal for each channel, which needs $2LM\log(2L)$
floating-point operations. For a specific frequency $f_l$,
$\boldsymbol{R}_y(f_l)$ can be estimated as
$\frac{1}{P}\sum_{p=1}^{P}\boldsymbol{y}_p(f_l)\boldsymbol{y}_p(f_l)^H$,
where $\boldsymbol{y}_p(f_l)\in\mathbb{C}^{M\times 1}$ is a sample
vector obtained according to (\ref{yf-definition}). It can be
verified that calculating $\boldsymbol{R}_y(f_l),\forall
l=\{1,\dots,2L\}$ needs $2M^2PL$ floating-point operations. Note
that the pseudo-inverse of
$(\boldsymbol{A}^{*}\odot\boldsymbol{A})$ can be computed
off-line. Thus we only need to perform a matrix-vector product to
compute $\boldsymbol{r}_x(f_l)$ according to (\ref{eqn-6}), which
requires $M^2 N$ floating-point operations. Consequently,
calculating $\boldsymbol{r}_x(f_l),\forall l=\{1,\dots,2L\}$ needs
$2M^2 NL$ floating-point operations. Overall, to obtain a spectrum
resolution of $1/(2NLT)$, the frequency-domain approach involves
$2MLP\log(2L)+2M^2(N+P)L$ floating-point operations in total.
Since we usually have $M\ll L$, the frequency-domain approach has
a much lower computational complexity compared with the
time-domain method. Note that we need $M^2\geq N$ to ensure that
(\ref{eqn4}) is invertible. Hence we have $2M^2(N+P)L>2N^2 L$,
which implies that the computational complexity of the
frequency-domain approach scales polynomially with the
downsampling factor $N$. This makes it unsuitable for the high
compression ratio scenario.

\begin{figure}
    \centering
    \includegraphics[width=8cm]{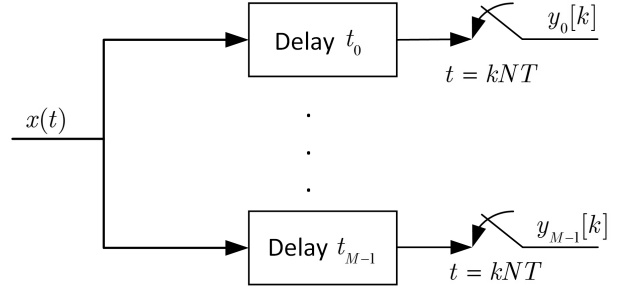}
    \caption{Compressed sampling architecture: multicoset sampling scheme}
    \label{fig3}
\end{figure}

\section{Proposed Fast Compressed Power Spectrum Estimation
Method} \label{sec:proposed-method} As discussed in our previous
section, current compressed power spectrum estimation methods,
particularly the time-domain approach, have a computational
complexity that might be too excessive for real-time spectrum
sensing systems. To address this issue, in this section, we
develop a computationally efficient compressed power spectrum
estimation method. In our proposed method, a multicoset sampling
scheme is employed to collect sub-Nyquist data samples.

Before proceeding, we first introduce the multicoset sampling
scheme. Similar to the AIC and the MWC, the multicoset sampling
architecture consists of a number of sampling channels, also
referred to as branches (see Fig. \ref{fig3}). In the $m$th
branch, the analog signal $x(t)$ is delayed by an amount of time,
$\Delta_m T$, and then sampled by a synchronized low-rate ADC
whose sampling interval is set to $NT$, where $T$ is the Nyquist
sampling interval, and $\Delta_m$ is set to be an integer smaller
than the downsampling factor $N$. Clearly, for the multicoset
sampling scheme, the output of the $m$th branch at the $l$th
sampling instant can be written as
\begin{align}
y_m[l]=x(lNT+\Delta_{m}T)=x[lN+\Delta_{m}],
\end{align}
where $x[n]=x(nT)$ denotes the Nyquist sample of $x(t)$. Define
$\boldsymbol{x}[l]\triangleq[x[lN]\phantom{0} \dots\phantom{0}
x[(l+1)N-1]]^T$, and $\boldsymbol{y}[l]\triangleq
[y_1[l]\phantom{0} \dots\phantom{0}y_{M}[l]]^T$. Then the data
sample collected at the $l$th sampling time instant is given by
\begin{align}
\boldsymbol{y}[l]=\boldsymbol{C}\boldsymbol{x}[l],
\label{mc-input-output-relation}
\end{align}
where $\boldsymbol{C}\in\{0,1\}^{M\times N}$ is a selection matrix
with only one nonzero entry in each row. We see that
(\ref{mc-input-output-relation}) is similar to
(\ref{AIC-input-output-relation}), except that $\boldsymbol{C}$ is
defined in a different way. We can follow the approach discussed
in Section \ref{sec:review} to reconstruct the power spectrum of
the analog signal $x(t)$. But such an approach has a prohibitively
high computational complexity. Notice that the data samples
obtained via the multicoset sampling scheme are a subset of the
Nyquist samples. As will be shown later, this property allows us
to establish an amiable relationship between the autocorrelation
sequence and the sub-Nyquist samples, based on which a fast
compressed power spectrum estimation method can be developed. This
is the reason why we use the multicoset sampler to collect
sub-Nyquist samples for our proposed method.

We see that the data samples $\{y_m[l]\}_{m=1,l=0}^{M,L-1}$
obtained via the multicoset sampling scheme are a subset of the
Nyquist samples $\{x[n]\}_{n=0}^{LN-1}$. To establish the
connection between the sub-Nyquist samples and Nyquist samples, we
define a data sequence $\{h[n]\}_{n=0}^{LN-1}$ and an indicator
sequence $\{I[n]\}_{n=0}^{LN-1}$, respectively, as follows
\begin{align}
h[n]=\begin{cases}
y_m[l],& n=lN+\Delta_{m}\\
0,&\text{otherwise}
\end{cases},
\end{align}
and
\begin{align}
I[n]=\begin{cases}
1,& n=lN+\Delta_{m}\\
0,&\text{otherwise}
\end{cases}.
\end{align}
It is easy to verify that
\begin{align}
h[n]=x[n]I[n].
\end{align}
We now show how to use the sequences $\{h[n]\}$ and $\{I[n]\}$ to
estimate the power spectrum of $\{x[n]\}$. A widely-used unbiased
estimate of the autocorrelation
$r_x[k]=\mathbb{E}[x[n]x^{*}[n-k]]$ is given as
\begin{align}
r_x[k]\approx\frac{1}{|\mathbb{Q}_k|}\sum_{n\in\mathbb{Q}_k}(x[n]x^{*}[n-k])
\quad k\in [-LN+1,LN-1], \label{eqn6}
\end{align}
where $\mathbb{Q}_k\triangleq\{n|0\le n-k\leq LN-1, 0\le n\leq
LN-1\}$, and $|\mathbb{Q}_k|$ denotes the size of $\mathbb{Q}_k$.
Since we only have access to $\{h[n]\}$, we propose a new unbiased
estimate of $\{r_x[k]\}$ which is given as
\begin{align}
r_x[k]\approx&\frac{1}{Q_k}\sum_{n\in\mathbb{\hat{Q}}_k}(x[n]x^{*}[n-k])\nonumber\\
=&\frac{1}{Q_k}\sum_{n\in\mathbb{Q}_k}(h[n]h^{*}[n-k]),
\label{eqn7}
\end{align}
where $\mathbb{\hat{Q}}_k\triangleq\{n|I[n]I[n-k]=1\}$ and
$Q_k\triangleq|\mathbb{\hat{Q}}_k|$. The proposed estimator, which
directly uses the sub-Nyquist samples, is quite simple and easy to
understand. Once the autocorrelation sequence $\{r_x[k]\}$ is
obtained, its power spectrum can readily be given as the Fourier
transform of the autocorrelation sequence.

\subsection{Recovery Condition}
Note that to estimate the power spectrum via (\ref{eqn7}), we have
to make sure that $Q_k>0,\forall k\in [-LN+1,LN-1]$, i.e. for each
$k\in [-LN+1,LN-1]$, the set
$\mathbb{\hat{Q}}_k=\{n|I[n]I[n-k]=1\}$ is a non-empty set. To
ensure this condition is satisfied, the time delays of the
multicoset sampling scheme have to be carefully devised. We have
the following result regarding the choice of the time delays
$\{\Delta_m\}$.

\newtheorem{lemma}{Lemma}
\begin{lemma}
For any integer $n$ satisfying $|n|\le\lfloor\frac{N}{2}\rfloor$,
where $\lfloor x\rfloor$ is the floor function that outputs the
greatest integer less than or equal to $x$, if there exist $m_1,
m_2\in\{1,\ldots,M\}$ and $c\in\{-1,0,1\}$ such that
\begin{align}
n=\Delta_{m_1}-\Delta_{m_2}+cN, \label{multicoset-condition}
\end{align}
then we have $Q_k>0$ for all $k\in [-LN+1,LN-1]$. \label{lemma1}
\end{lemma}
\begin{proof}
See Appendix \ref{appA}.
\end{proof}

\emph{Remarks:} We note that the recovery condition
(\ref{multicoset-condition}) is slightly relaxed than the sparse
ruler condition given in \cite{ArianandaLeus12}, and is equivalent
to the recovery condition given in \cite{RomeroLeus13} which is
termed as the circular sparse ruler. For the sparse ruler
condition \cite{ArianandaLeus12}, it requires the set
$\{0,1,\dots,\lfloor N/2\rfloor\}$ to be a subset of
$\{|\Delta_{m_1}-\Delta_{m_2}|:0<m_1\le M, 0<m_2\le M\}$. Such a
condition amounts to saying that for any integer $|n|\leq\lfloor
N/2\rfloor$, their exist $m_1$ and $m_2$ such that
\begin{align}
n=\Delta_{m_1}-\Delta_{m_2}. \label{multicoset-condition-2}
\end{align}
We see that the condition (\ref{multicoset-condition-2}) is the
same as (\ref{multicoset-condition}) if we set $c=0$. This means
our condition (\ref{multicoset-condition}) is more relaxed than
the condition discussed in \cite{ArianandaLeus12}. For example,
let $N=8$, $M=4$. Then we have two different delay sets
$\{\Delta_{m}\}=\{0,2,3,4\}$ and $\{\Delta_{m}\}=\{0,3,5,7\}$,
both of which satisfy our condition (\ref{multicoset-condition}).
Nevertheless, only the first delay set satisfies
(\ref{multicoset-condition-2}). Finding a smallest number of
branches, $M$, such that (\ref{multicoset-condition}) can be met
is referred to as the circular sparse ruler problem and has been
investigated in previous studies, e.g.
\cite{GonzalezDominguez14,DominguezGonzalez14}. Specifically, in
\cite{YenTsai13}, it was shown that if $M\ge(N+1)/2$, then one can
always find a multicoset sampling pattern satisfying
(\ref{multicoset-condition-2}), and consequently
(\ref{multicoset-condition}). Furthermore,
\cite{GonzalezDominguez14,DominguezGonzalez14} proposed some
special universal sampling patterns for the scenario where
$M<(N+1)/2$. Note that the condition (\ref{multicoset-condition})
guarantees the recovery of the power spectrum, without placing any
sparsity constraint on the spectrum under monitoring. In contrast,
the earlier work \cite{FengBresler96} discussed a universal
sampling pattern that ensures perfect recovery of a multi-band
signal with a given spectral occupancy bound.

\subsection{Proposed Algorithm}
Note that directly calculating $\{r_x[k]\}_{k=-LN+1}^{LN-1}$ via
(\ref{eqn7}) will need a total number of $2L^2N^2+LN$
floating-point operations. Such a computational complexity would
become unacceptable when $L$ is large, which is usually the case
in order to provide a fine spectrum resolution. We now show how to
use FFT to reduce the computational complexity of the proposed
estimator (\ref{eqn7}). Define
\begin{align}
r_h[k]=\sum_{n\in\mathbb{Q}_k}(h[n]h^{*}[n-k]), \label{eqn13}
\end{align}
where $\mathbb{Q}_k=\{n|0\le n-k\leq LN-1, 0\le n\leq LN-1\}$. The
above equation can be rewritten as
\begin{align}
r_h[k]=\sum_{n=0}^{LN-1}(h[n]h^{*}[n-k]),
\end{align}
by setting $h[n]=0$ for $n<0$ and $n\ge LN$. We define a new
sequence $\{\bar{h}[n]\}_{n=-LN+1}^{LN-1}$ as
\begin{align}
\bar{h}[n]=\begin{cases}
h[n],& LN-1\geq n\ge0\\
0,&-LN+1\leq n<0
\end{cases},
\end{align}
and let $\{\hat{h}[n]\}_{n=-LN+1}^{LN-1}$ be the reverse of
$\{\bar{h}[n]\}_{n=-LN+1}^{LN-1}$, i.e
\begin{align}
\hat{h}[n]=\bar{h}[-n].
\end{align}
Then we have
\begin{align}
r_h[k]&=\begin{cases}
\sum_{n=-LN+1+k}^{LN-1}(\bar{h}[n]\hat{h}^{*}[k-n]),&k\ge0\\
\sum_{n=-LN+1}^{LN-1+k}(\bar{h}[n]\hat{h}^{*}[k-n]),&k<0\\
\end{cases}\nonumber\\
&\stackrel{(a)}{=}\sum_{n=-LN-1}^{LN-1}(\bar{h}[n]\hat{h}_{P}^{*}[k-n])\nonumber\\
&=(\bar{h}\star\hat{h}^{*})[k],
\end{align}
where $P\triangleq 2NL-1$, $\hat{h}_{P}[n]$ is a periodic
summation of $\hat{h}[n]$ defined as
\begin{align}
\hat{h}_{P}[n]\triangleq\sum_{k=-\infty}^{+\infty}\hat{h}[n-kP],
\end{align}
$(a)$ comes from the fact that $\bar{h}[n]=0$ for $-LN+1\le n<0$
and $\hat{h}_{P}[n]=\hat{h}[n-NL]=0$ for $-2NL+1< n\le -NL$, and
the symbol $\star$ in the last equality denotes the circular
convolution.

Define
\begin{align}
\boldsymbol{r}_h&\triangleq[r_h[-LN+1]\phantom{0}\dots\phantom{0}r_h[LN-1]]^T,\nonumber\\
\boldsymbol{\bar{h}}&\triangleq[\bar{h}[-LN+1]\phantom{0}\dots\phantom{0}\bar{h}[LN-1]]^T,\nonumber\\
\boldsymbol{\hat{h}}&\triangleq[\hat{h}[-LN+1]\phantom{0}\dots\phantom{0}\hat{h}[LN-1]]^T.\nonumber
\end{align}
Invoking the circular convolution theorem, we have
\begin{align}
\boldsymbol{F}_{2NL-1}\boldsymbol{r}_h=
(\boldsymbol{F}_{2NL-1}\boldsymbol{\bar{h}})\circ(\boldsymbol{F}_{2NL-1}\boldsymbol{\hat{h}}),
\end{align}
where $\boldsymbol{F}_{2NL-1}$ and $\circ$ denote the
$(2NL-1)$-point discrete Fourier transform (DFT) matrix and the
element-wise product, respectively.

As the sequence $\{\hat{h}[n]\}$ is the time reversal of the
sequence $\{\bar{h}[n]\}$, according to the time reversal and
complex-conjugate properties of DFT, the DFT of the sequence
$\{\hat{h}^{*}[n]\}$ is the complex conjugate of the DFT of the
sequence $\{\bar{h}[n]\}$ \cite{Proakis01}. Therefore we have
\begin{align}
(\boldsymbol{F}_{2NL-1}\boldsymbol{\bar{h}})\circ(\boldsymbol{F}_{2NL-1}
\boldsymbol{\hat{h}})=|\boldsymbol{F}_{2NL-1}\boldsymbol{\bar{h}}|^2,
\end{align}
where $|\cdot|^2$ denotes the element-wise square modulus of a
complex vector. Thus $\boldsymbol{r}_h$ can be computed as
\begin{align}
\boldsymbol{r}_h=\boldsymbol{F}^{-1}_{2NL-1}|\boldsymbol{F}_{2NL-1}\boldsymbol{\bar{h}}|^2.
\label{eqn11}
\end{align}
By resorting to the fast Fourier transform (FFT),
$\boldsymbol{r}_h$ can be efficiently calculated.

Recalling that $Q_k=|\mathbb{\hat{Q}}_k|$ and
$\mathbb{\hat{Q}}_k=\{n|I[n]I[n-k]=1\}$, $Q_k$ can be expressed as
\begin{align}
Q_k=\sum\limits_{n\in\mathbb{Q}_k}I[n]I[n-k] \label{eqn12}
\end{align}
Notice that (\ref{eqn12}) has a form similar to (\ref{eqn13}).
Therefore by following the same approach of calculating
(\ref{eqn13}), $\{Q_k\}$ can also be efficiently computed via
FFTs. Specifically, define
\begin{align}
\boldsymbol{q}&\triangleq[Q_{-LN+1}\phantom{0}\dots\phantom{0}Q_{LN-1}]^T,\nonumber\\
\boldsymbol{\bar{I}}&\triangleq[\bar{I}[-LN+1]\phantom{0}\dots\phantom{0}\bar{I}[LN-1]]^T,\nonumber
\end{align}
where
\begin{align}
\bar{I}[n]=\begin{cases}
I[n],& LN-1\geq n\ge0\\
0,&-LN+1\leq n<0
\end{cases}.
\end{align}
Then $\boldsymbol{q}$ can be computed via
\begin{align}
\boldsymbol{q}=\boldsymbol{F}^{-1}_{2NL-1}|\boldsymbol{F}_{2NL-1}\boldsymbol{\bar{I}}|^2.
\label{q-sequence}
\end{align}

After obtaining the sequences $\{Q_k\}$ and $\{r_h[k]\}$, the
power spectrum of the original signal $x(t)$ can be readily
estimated by taking the DFT of $\{r_x[k]\}$, in which we have
$r_x[k]=r_h[k]/Q_k$ according to (\ref{eqn7}). For clarity, we
plot the block diagram of our proposed wideband power spectrum
estimation method in Fig.\ref{fig4}.

\begin{figure*}[t]
    \centering
    \includegraphics[width=16cm]{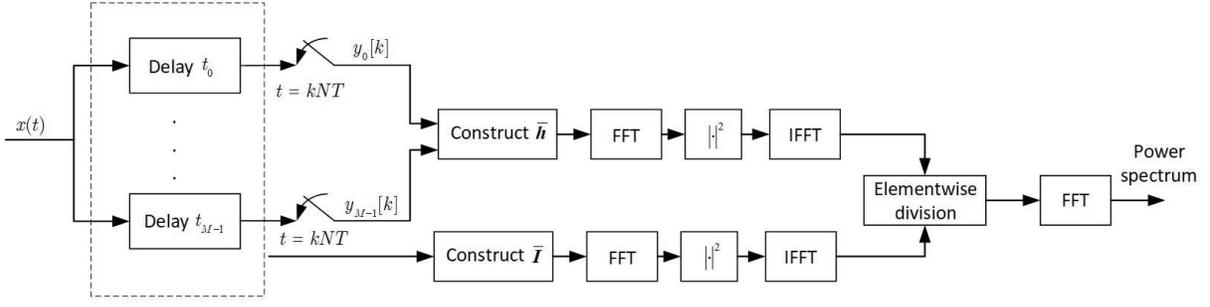}
    \caption{Block diagram of the proposed compressed power spectrum
        estimation method}
    \label{fig4}
\end{figure*}

\subsection{Computational Complexity}
We see that our proposed compressed power spectrum estimation
method only involves FFT/IFFT operations and some simple
multiplication calculations. More precisely, to obtain the power
spectrum of $x(t)$ with a spectrum resolution of $1/((2NL-1)T)$,
we only need to calculate $\boldsymbol{r}_h$ and $\boldsymbol{q}$
via (\ref{eqn11}) and (\ref{q-sequence}), respectively, and then
compute the FFT of the sequence
$\boldsymbol{r}_x=\boldsymbol{r}_h./\boldsymbol{q}$, where $./$
stands for the element-wise division. Note that $\boldsymbol{q}$
can be pre-calculated since it is only dependent on the sampling
pattern of the multi-coset scheme. Therefore the proposed method
only requires to execute the FFT of a $(2LN-1)$-point sequence
three times, plus $2LN-1$ multiplication calculations. It can be
easily checked that the proposed method involves $(6LN-3)\log
(2LN-1)+2LN-1$ floating-point operations in total, which scales
linearly with the number of samples (in time) $L$ and the
downsampling factor $N$. Clearly, our proposed method has a lower
computational complexity than existing methods that have a
complexity either scaling polynomially with $L$ or scaling
polynomially with $N$. In addition, an important advantage of our
proposed method over existing methods is that our proposed method,
where FFT is the computationally dominant task, can be efficiently
computed in a parallel manner. Such a merit makes it possible to
develop a practical system to perform real-time wideband spectrum
sensing. According to \cite{FFTIP}, a high-performance FPGA such
as Xilinx Virtex-6 can complete a complex $32000$-point FFT within
$0.4$ millisecond (ms). Since our proposed method involves three
sequential FFT operations, it can reconstruct the power spectrum
of a frequency band of 1GHz with a reasonably fine spectrum
resolution of $62.5$kHz within $1.2$ms, which meets most real-time
sensing applications.

\subsection{Performance Analysis}
In this subsection, we provide a theoretical analysis of the mean
square error (MSE) of the proposed power spectrum method.
Specifically, the MSE can be calculated as
\begin{align}
\text{MSE}=\mathbb{E}[\|\boldsymbol{\hat{s}}-\boldsymbol{s}\|_2^2],
\end{align}
where $\boldsymbol{\hat{s}}$ and $\boldsymbol{s}$ denote the
estimated power spectrum and the true one, respectively. The MSE
can be further expressed as
\begin{align}
\mathbb{E}[\|\boldsymbol{\hat{s}}-\boldsymbol{s}\|_2^2]=&\mathbb{E}[\|\boldsymbol{\hat{s}}\|_2^2]-
2\mathbb{E}[\boldsymbol{\hat{s}}]^T\boldsymbol{s}+\|\boldsymbol{s}\|_2^2\nonumber\\
=&\mathbb{E}[\|\boldsymbol{F}\boldsymbol{r}_x\|_2^2]-2\mathbb{E}[\boldsymbol{\hat{s}}]^H\boldsymbol{s}
+\|\boldsymbol{s}\|_2^2\nonumber\\
=&\mathbb{E}[\|\boldsymbol{r}_x\|_2^2]-2\mathbb{E}[\boldsymbol{\hat{s}}]^H\boldsymbol{s}+\|\boldsymbol{s}\|_2^2\nonumber\\
\stackrel{(a)}{=}&\|\mathbb{E}[\boldsymbol{r}_x]\|_2^2+\boldsymbol{1}^T
\mathbb{D}(\boldsymbol{r}_x)-2\mathbb{E}[\boldsymbol{\hat{s}}]^H\boldsymbol{s}+\|\boldsymbol{s}\|_2^2\nonumber\\
=&\|\boldsymbol{F}\mathbb{E}[\boldsymbol{r}_x]\|_2^2+\boldsymbol{1}^T
\mathbb{D}(\boldsymbol{r}_x)-2\mathbb{E}[\boldsymbol{\hat{s}}]^H\boldsymbol{s}+\|\boldsymbol{s}\|_2^2\nonumber\\
=&\|\mathbb{E}[\boldsymbol{\hat{s}}]\|_2^2+\boldsymbol{1}^T
\mathbb{D}(\boldsymbol{r}_x)-2\mathbb{E}[\boldsymbol{\hat{s}}]^H\boldsymbol{s}+\|\boldsymbol{s}\|_2^2\nonumber\\
=&\|\mathbb{E}[\boldsymbol{\hat{s}}]-\boldsymbol{s}\|_2^2+\boldsymbol{1}^T\mathbb{D}(\boldsymbol{r}_x),
\end{align}
where $\boldsymbol{r}_x$ denotes the estimated autocorrelation
vector, $\boldsymbol{F}$ is the DFT matrix, and $(a)$ comes from
\begin{align}
\mathbb{E}[\|\boldsymbol{r}_x\|_2^2]
=\|\mathbb{E}[\boldsymbol{r}_x]\|_2^2+\boldsymbol{1}^T\mathbb{D}(\boldsymbol{r}_x),
\end{align}
in which $\mathbb{D}(\boldsymbol{r}_x)$ denotes the operation
taking the variance of each element of $\boldsymbol{r}_x$, and
$\boldsymbol{1}$ denotes the vector with all of its entries equal
to 1. Since our estimator (\ref{eqn7}) is an unbiased estimator,
we have $\mathbb{E}(\boldsymbol{\hat{s}})=\boldsymbol{s}$. The MSE
of our proposed estimator can therefore be given as
\begin{align}
\text{MSE}=\boldsymbol{1}^T\mathbb{D}(\boldsymbol{r}_x).
\end{align}
Thus we only need to evaluate the variance of the elements of
$\boldsymbol{r}_x$. For $r_x[k]$, its variance can be computed as
\begin{align}
&\mathbb{D}(r_x[k]) \nonumber\\
=&\mathbb{E}[|r_x[k]|^2]-|\mathbb{E}[r_x[k]]|^2\nonumber\\
=&\mathbb{E}\bigg[\bigg|\frac{1}{Q_k}\sum_{n\in\mathbb{Q}_k}(h[n]h^{*}[n-k])\bigg|^2\bigg]-
|\mathbb{E}[r_x[k]]|^2\nonumber\\
=&\mathbb{E}\bigg[\frac{1}{Q_k^2}\sum_{n\in\mathbb{Q}_k}
\sum_{m\in\mathbb{Q}_k}(h[n]h^{*}[n-k]h^{*}[m]h[m-k])\bigg]
-|\mathbb{E}[r_x[k]]|^2\nonumber\\
=&\mathbb{E}\bigg[\frac{1}{Q_k^2}\sum_{n\in
\hat{\mathbb{Q}}_{k},m\in\hat{\mathbb{Q}}_{k}}(x[n]x^{*}[n-k]x^{*}[m]x[m-k])\bigg]
-|\mathbb{E}[r_x[k]]|^2.
\end{align}
Observe that calculation of the variance of $r_{x}[k]$ is not
trivial as it involves fourth-order moments of the multi-band
signal $x[n]$, which requires the knowledge of the distribution of
$x[n]$. If $x[n]$ is Gaussian distributed, the fourth-order
moments can be simplified as a sum of products of second order
moments \cite{ArianandaLeus12}.

In the following, we consider the special case where $x(t)$ is a
temporally white signal with zero mean and variance $\sigma^2$.
When $k\ne0$, we have
\begin{align}
&\mathbb{D}(r_x[k]) \nonumber\\
=&\mathbb{E}\bigg[\frac{1}{Q_k^2}
\sum_{n\in\hat{\mathbb{Q}}_{k},m\in\hat{\mathbb{Q}}_{k}}(x[n]x^{*}[n-k]x^{*}[m]x[m-k])\bigg]-
|\mathbb{E}[r_x[k]]|^2\nonumber\\
=&\frac{1}{Q_k^2}\sum_{n\in\hat{\mathbb{Q}}_{k},m\in\hat{\mathbb{Q}}_{k}}
\mathbb{E}\left[x[n]x^{*}[m]x[n-k]x^{*}[m-k]\right]-|\mathbb{E}[r_x[k]]|^2\nonumber\\
=&\frac{1}{Q_k^2}\sum_{n\in\hat{\mathbb{Q}}_{k}}\mathbb{E}[x[n]x^{*}[n]]
\mathbb{E}[x[n-k]x^{*}[n-k]]-|\mathbb{E}[r_x[k]]|^2\nonumber\\
=&\frac{1}{Q_k^2}\sum_{n\in\hat{\mathbb{Q}}_{k}}\sigma^4\nonumber\\
=&\frac{1}{Q_k}\sigma^4.
\end{align}
When $k=0$, we have
\begin{align}
\mathbb{D}(r_x[0])=&\mathbb{E}\bigg[\frac{1}{Q_0^2}\sum_{n\in
\hat{\mathbb{Q}}_{0},m\in\hat{\mathbb{Q}}_{0}}(x[n]x^{*}[n]x^{*}[m]x[m])\bigg]-|\mathbb{E}[r_x[0]]|^2\nonumber\\
=&\mathbb{E}\bigg[\frac{1}{Q_0^2}\sum_{n\in\hat{\mathbb{Q}}_{0}}(x[n]x^{*}[n]x^{*}[n]x[n])\bigg]
\nonumber\\
&+\mathbb{E}\bigg[\frac{1}{Q_0^2}\sum_{n\in\hat{\mathbb{Q}}_{0},m\in\hat{\mathbb{Q}}_{0},m\ne
n}
(x[n]x^{*}[n]x^{*}[m]x[m])\bigg]-\sigma^4\nonumber\\
=&\frac{1}{Q_0^2}\sum_{n\in\hat{\mathbb{Q}}_{0}}\mathbb{E}\left[x[n]x^{*}[n]x^{*}[n]x[n]\right]
\nonumber\\
&+\frac{1}{Q_0^2}\sum_{n\in\hat{\mathbb{Q}}_{0},m\in\hat{\mathbb{Q}}_{0},m\ne
n}
\mathbb{E}[x[n]x^{*}[n]]\mathbb{E}[x^{*}[m]x[m]]-\sigma^4\nonumber\\
=&\frac{2\sigma^4}{Q_0}+\frac{1}{Q_0^2}(\sigma^4(Q_0^2-Q_0))-\sigma^4\nonumber\\
=&\frac{\sigma^4}{Q_0}.
\end{align}
Therefore the MSE of the proposed method for this special scenario
is given by
\begin{align}
\text{MSE}=\boldsymbol{1}^T\mathbb{D}(\boldsymbol{r}_x)=\sigma^4\sum_k1/Q_k.
\end{align}
We see that the MSE of the proposed method is related to the
number of collected data samples as well as the sampling pattern.
It is easy to see that more data samples lead to a lower MSE.
Also, given a fixed downsampling factor $N$, increasing the number
of branches $M$ results in larger $\{Q_k\}$ and thus a higher
estimation accuracy.

\subsection{Discussions} \label{sec:proposed-method-discussion}
Our proposed method is based on the assumption that the multi-band
signal $x(t)$ is wide-sense stationary. Note that wide-sense
stationarity is an assumption widely adopted for compressed power
spectrum estimation and spectrum sensing, e.g.
\cite{ArianandaLeus12,YenTsai13,CohenEldar14}. Nevertheless, in
practice, the signal of interest might be nonstationary or
cyclostationary. For nonstationary signals, they may have slowly
time-varying statistics. In this case, they can be treated as
wide-sense stationary signals within a sufficiently short period
of time. Also, our proposed method is applicable to communication
signals which are known to be cyclostationary. To explains this,
let $r(t,\tau)=E[x(t)x(t-\tau)]$ denote the autocorrelation of a
random process $x(t)$. For cyclostationary signals, $r(t,\tau)$ is
cyclic in $t$ and can be expanded in Fourier series:
\begin{align}
r(t,\tau)=\sum_{\alpha=-\infty}^{+\infty}r^{\alpha}(\tau)e^{j2\pi\alpha
t},
\end{align}
where $r^{\alpha}(\tau)$ is called the cyclic autocorrelation
function defined as
\begin{align}
r^{\alpha}(\tau)=\int_{-\infty}^{+\infty}r(t,\tau)e^{-j2\pi\alpha
t}dt.
\end{align}
The cyclic spectrum $S(\alpha,f)$ used to analyze the
cyclostationary signal can be calculated as the Fourier transform
of the cyclic autocorrelation function at cyclic frequency
$\alpha$, i.e.
\begin{align}
S(\alpha,f)=&\int_{-\infty}^{+\infty}r^{\alpha}(\tau)e^{-j2\pi f
\tau}d\tau \nonumber\\
=&
\int_{-\infty}^{+\infty}\int_{-\infty}^{+\infty}r(t,\tau)e^{-j2\pi
f \tau}e^{-j2\pi\alpha t}d\tau d t.
\end{align}
The cyclic spectrum at zeroth cyclic frequency ($\alpha=0$), also
called average power spectral density, is therefore given as
\begin{align}
S(0,f)=&\int_{-\infty}^{+\infty}\int_{-\infty}^{+\infty}r(t,\tau)e^{-j2\pi
f \tau} d t d\tau \nonumber\\
=&\int_{-\infty}^{+\infty}\left(\int_{-\infty}^{+\infty}r(t,\tau)d t\right)e^{-j2\pi f \tau}d\tau\nonumber\\
=&\int_{-\infty}^{+\infty}\hat{r}(\tau)e^{-j2\pi f \tau} d t
d\tau, \label{eqn14}
\end{align}
where
\begin{align}
\hat{r}(\tau)=\int_{-\infty}^{+\infty}r(t,\tau)d t.
\end{align}
From (\ref{eqn14}), we see that the average power spectrum of a
cyclostationary signal can be calculated as the Fourier transform
of $\hat{r}(\tau)$. Note that the autocorrelation $r_x[k]$
estimated via (\ref{eqn6}) or (\ref{eqn7}) can actually be
considered as an estimate of $\hat{r}(\tau)$ for cyclostationary
signals. Therefore our proposed method is applicable to
cyclostationary signals and renders an estimate of its average
power spectrum, i.e. its cyclic spectrum at zeroth cyclic
frequency.

\begin{figure*}[t]
    \centering
    \subfloat[]{
        \includegraphics [width=200pt]{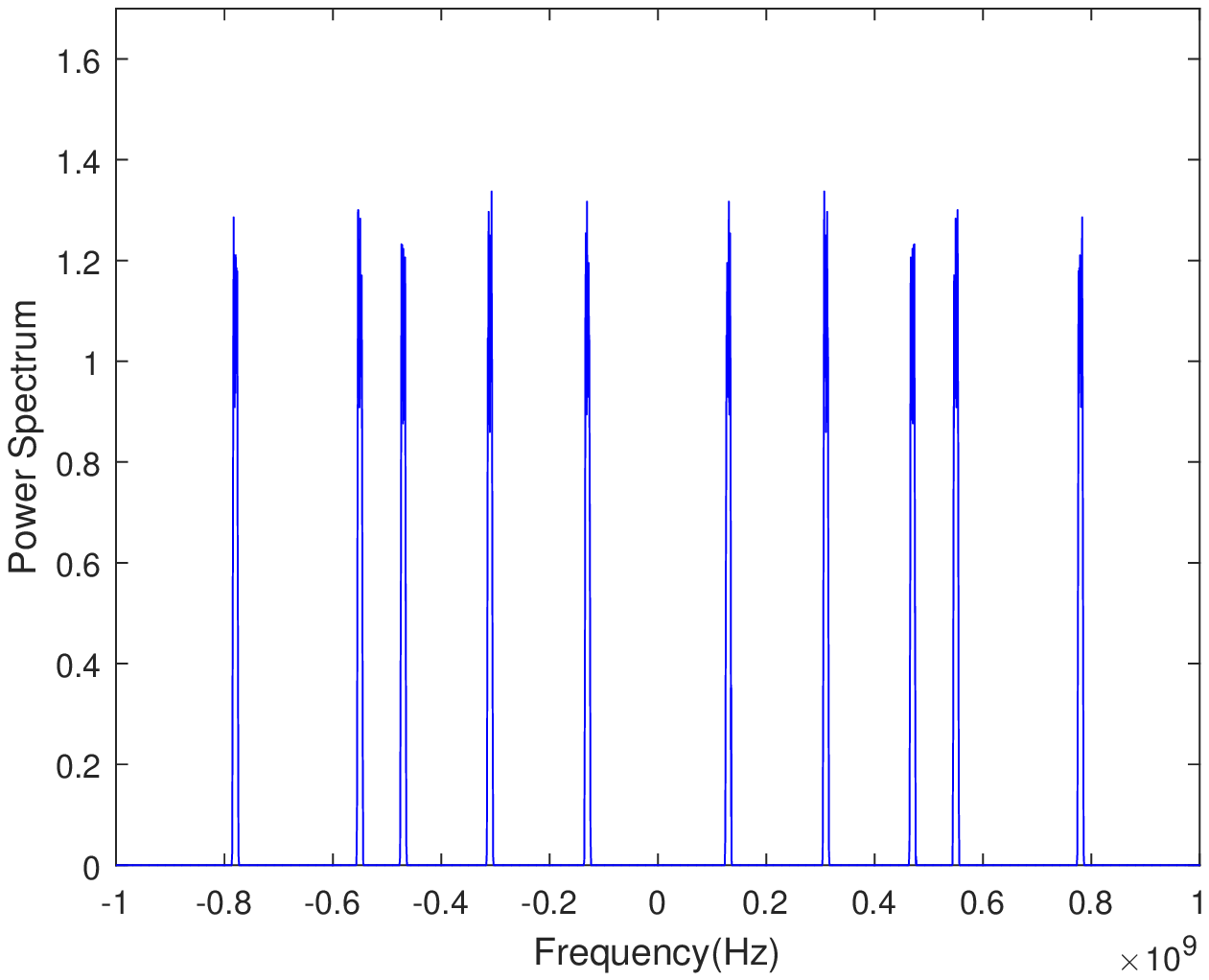}
    }
    \hspace{20pt}
    \subfloat[]{
        \includegraphics [width=200pt]{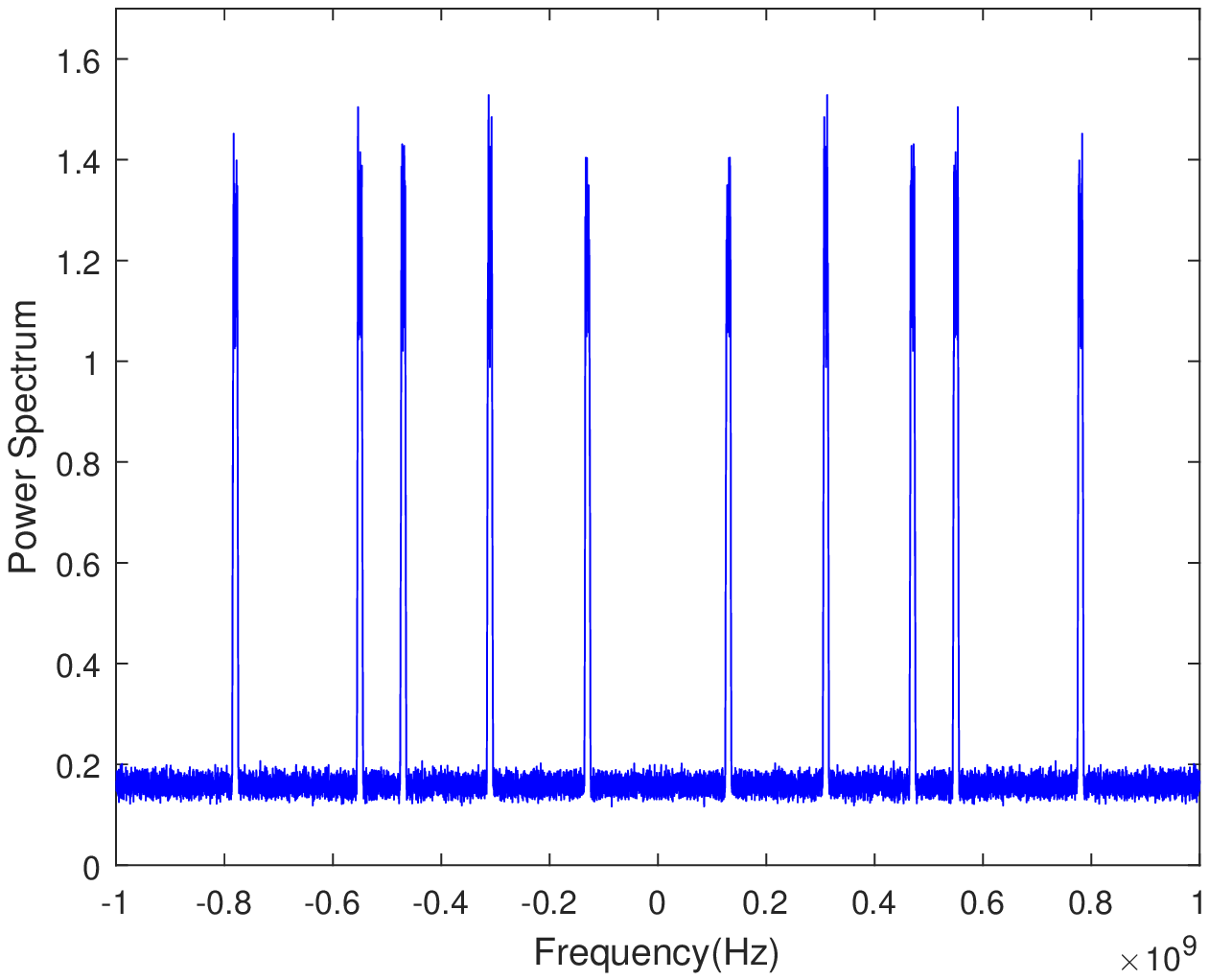}
    }
    \\
    \subfloat[]{
        \includegraphics [width=200pt]{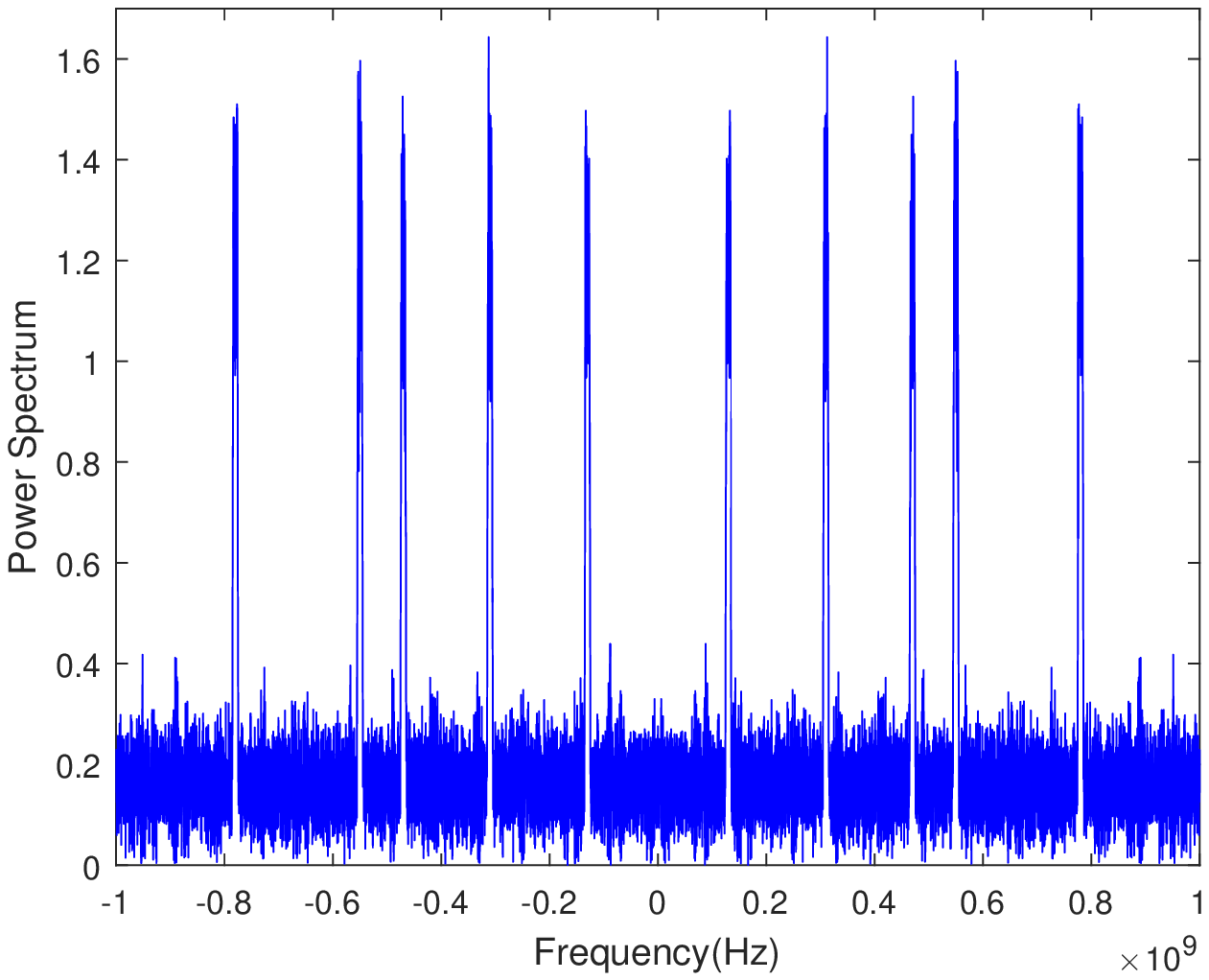}
    }
    \hspace{20pt}
    \subfloat[]{
        \includegraphics [width=200pt]{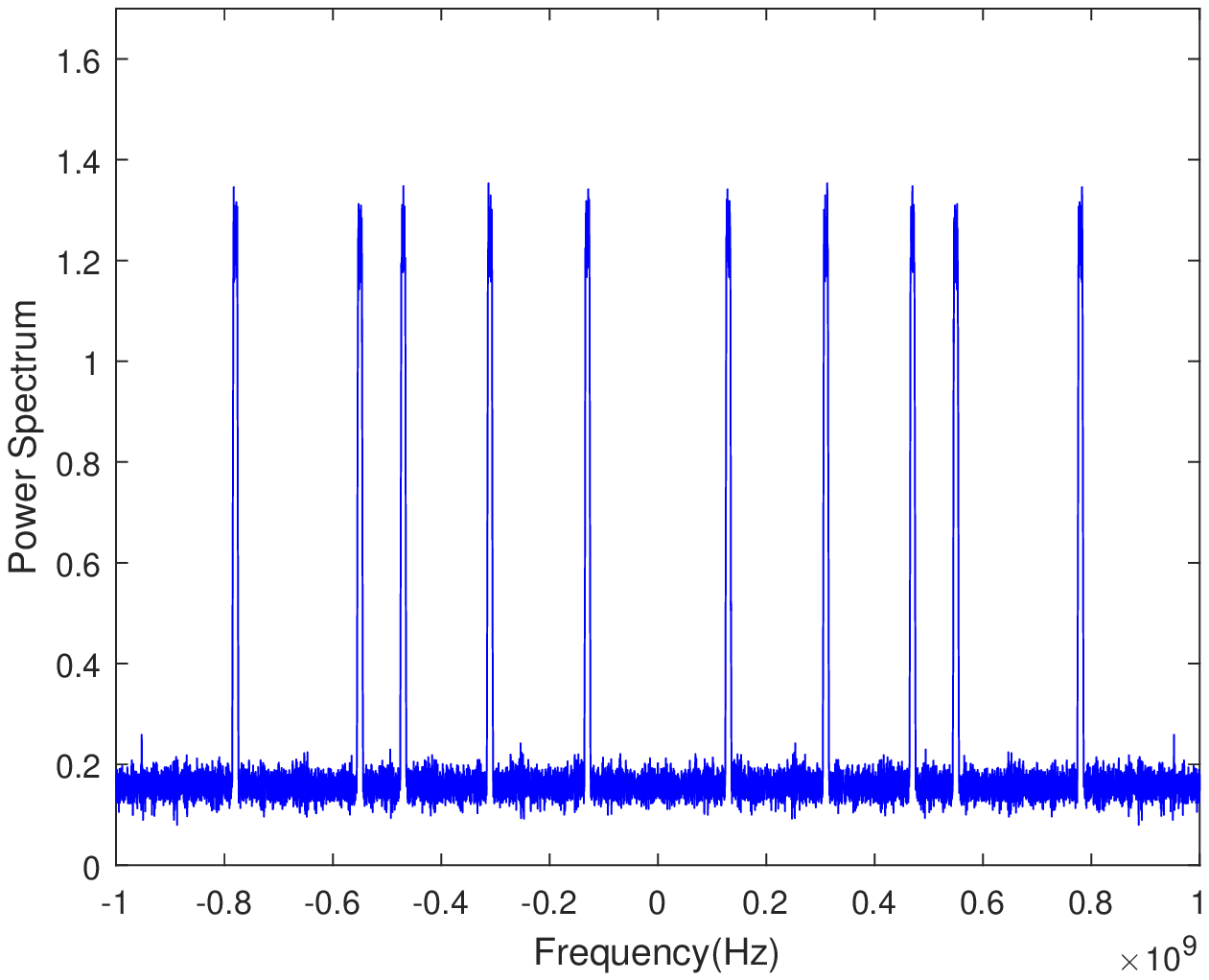}
    }
    \caption{(a) Power spectrum reconstructed using noiseless Nyquist samples collected within 1ms;
    (b) Power spectrum reconstructed using noisy Nyquist samples collected within 1ms;
        (c) Power spectrum reconstructed via our proposed method
        with 1ms noisy sub-Nyquist samples; (d) Power spectrum
        reconstructed via our proposed method with 10ms noisy sub-Nyquist samples}
    \label{fig5}
\end{figure*}

\section{Simulation Results} \label{sec:simulation-results}
In this section, we carry out experiments to show the efficiency
and effectiveness of the proposed compressed power spectrum
estimation method. In the following, we mainly report results on
wideband spectrum sensing of wide-sense stationary signals.
Results on cyclostationary signals are also included in order to
corroborate our claim and analysis in Section
\ref{sec:proposed-method-discussion}.

\subsection{Results on Wide-Sense Stationary Signals}
In our experiments, we generate $5$ wide-sense stationary signals
within the frequency range $[0,1]$GHz by filtering the zero-mean
unit-variance Gaussian white noise using band-pass FIR filters.
Our objective is to identify the frequency locations of the
signals that spread over the frequency band $[0,1]$GHz. Clearly,
the Nyquist sampling rate is $2$GHz. The bandwidth of each signal
component is set to $10$MHz and their carrier frequencies are set
to $130$, $310$, $470$, $550$, and $780$MHz, respectively. For our
proposed method, a multicoset sampling scheme is employed to
collect sub-Nyquist data samples, in which the number of sampling
channels is set to $M=8$, the sampling rate for each channel is
set to $80$MHz and the time delays are set to $\{0, 0.5, 1, 1.5,
2, 2.5, 3, 6.5\}$ nanoseconds (ns). The downsampling factor is
equal to $N=f_{nyq}/(80\text{MHz})=25$. The spectrum resolution is
set to $62.5$kHz, which corresponds to a power spectrum of length
$32000$. Hence, we have $2NL\ge32000$ and each sampling channel
needs to collect $L\geq 32000/2N=640$ data samples. In our
experiments, we collect data samples of $1$ms' duration for each
channel. We then calculate the autocorrelation sequence and
truncate it to a vector of length $32000$. For our proposed method
and the conventional time-domain algorithm \cite{ArianandaLeus12},
a hamming window is added to the estimated autocorrelation
sequence to enhance the power spectrum estimation accuracy.

Fig. \ref{fig5} plots the power spectrum reconstructed using
noiseless Nyquist data samples collected within $1$ms, using noisy
Nyquist data samples collected within $1$ms, the power spectrum
reconstructed via our proposed method with noisy sub-Nyquist data
samples collected within $1$ms, and noisy sub-Nyquist data samples
collected within $10$ms, where the signal to noise ratio (SNR) is
set to -5dB for the noisy scenario. The noisy signal is generated
by corrupting the original signal $x(t)$ with zero mean Gaussian
noise. The SNR is defined as
\begin{align}
\text{SNR}=10\log_{10}\frac{\sum_{n=1}^{N_t}
|x[n]|^2}{N_t\sigma^2},
\end{align}
where $\{x[n]\}$ denote the Nyquist samples of $x(t)$, $N_t$ is
the number of the Nyquist samples, and $\sigma^2$ denotes the
variance of the Gaussian noise. From Fig. \ref{fig5}, we see that
our proposed method is able to accurately recover the true power
spectrum. Specifically, the normalized mean square errors (NMSE),
$E[\|\boldsymbol{s}-\boldsymbol{\hat{s}}\|_2^2/\|\boldsymbol{s}\|_2^2]$,
for Fig. \ref{fig5}(b)--(d) are respectively given as $0.0057$,
$0.0377$, and $0.0037$, where $\boldsymbol{s}$ and
$\boldsymbol{\hat{s}}$ denotes the groundtruth and estimated power
spectrum, respectively. Moreover, we observe that, with
sub-Nyquist samples collected within $10$ms, our proposed method
provides an accuracy slightly higher than that obtained using
$1$ms Nyquist samples. This result implies that the performance
loss caused by downsampling can be compensated by increasing the
sampling time.

\begin{figure}[t]
    \centering
    \includegraphics [width=200pt,height=180pt]{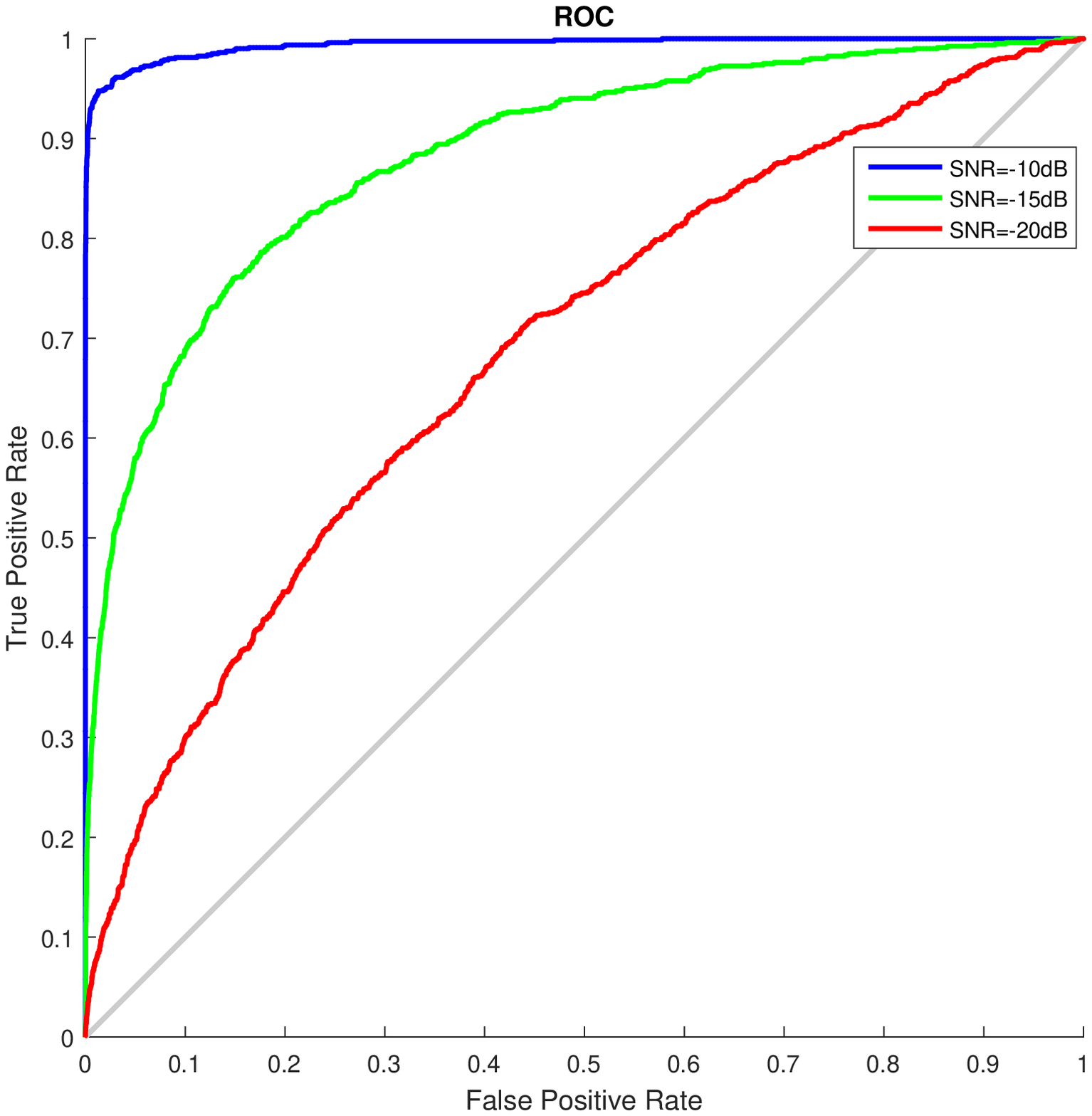}
    \caption{TPR vs. FPR for different SNRs.}
    \label{fig6}
\end{figure}

\begin{figure}[t]
    \centering
    \includegraphics [width=200pt,height=180pt]{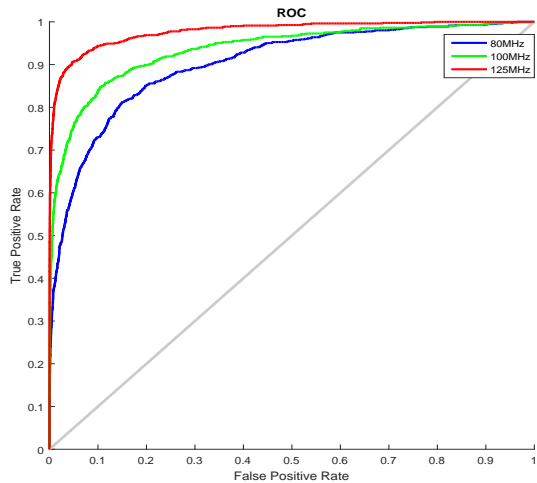}
    \caption{TPR vs. FPR under different compression ratios.}
    \label{fig7}
\end{figure}

\begin{figure}[t]
    \centering
    \includegraphics [width=200pt,height=180pt]{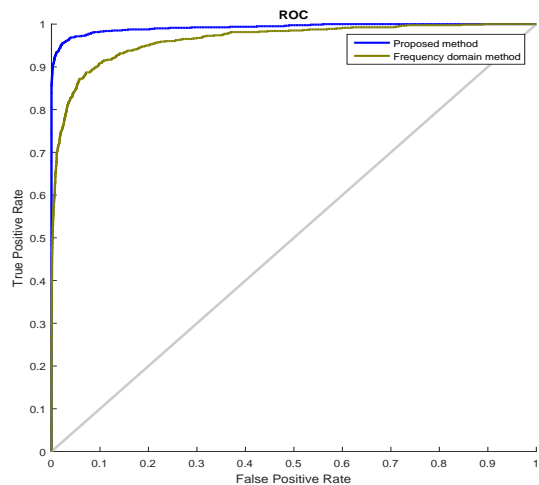}
    \caption{The ROC of our proposed method and the frequency-domain method. The run
        times of our proposed method and the frequency-domain method are $0.17s$ and $0.35s$, respectively.}
    \label{fig8}
\end{figure}

\begin{figure}[t]
    \centering
    \includegraphics [width=200pt,height=180pt]{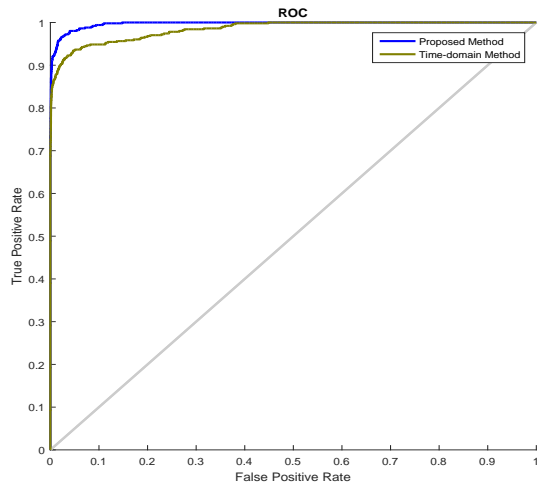}
    \caption{The ROC of our proposed method and the time-domain method. The run
        times of our proposed method and the time-domain method are $0.03s$ and $0.13s$, respectively.}
    \label{fig11}
\end{figure}

To more fully examine the performance, we plot the ROC curve for
our proposed method under different SNRs in Fig. \ref{fig6}. The
ROC curve is created by plotting the true positive rate (TPR)
against the false positive rate (FPR) at various threshold
settings. Specifically, given an estimated power spectrum, one can
use a threshold to identify the frequency locations of the signals
that spread over a wide frequency band, i.e. those grid points
which have a higher energy level than the threshold are considered
positive, otherwise it is considered that no signals reside on
those grid points. The TPR and the FPR can be calculated by
comparing the detection result and the groundtruth. The ROC curve
can then be obtained by trying different thresholds. From Fig.
\ref{fig6}, we see that our proposed method can achieve reliable
detection even in the low SNR regime, say,
$\text{SNR}=-10\text{dB}$. This result demonstrates the
superiority of the use of statistical information for spectrum
sensing. Such a merit is particularly useful for wideband spectrum
sensing because, to obtain a practically meaningful receiver
sensitivity (typical value for the receiver sensitivity is around
$-70\sim -90$dBm), the receiver has to operate in a low SNR
regime. To illustrate this, note that the receiver sensitivity can
be calculated as
\begin{align}
S=10\log(kT_{\text{syst}})+10\log(B)+NF_{RX}+\text{SNR},
\end{align}
where $10\log(kT_{\text{syst}})$ is equal to -174dBm/Hz for a
system temperature of 17 degrees in Celsius, $B$ is the bandwidth
of the signal in hertz, and $NF_{RX}$ is the noise figure of the
receiver in decibels whose typical value is 6dB \cite{NFTV}. Thus
the receiver sensitivity is given by
$S=-174\text{dBm/Hz}+10\log(1\text{GHz})+6\text{dB}+\text{SNR}=-78\text{dBm}+\text{SNR}$.
To achieve a receiver sensitivity of $-80$dBm, the SNR is equal to
$-80-(-78)=-2\text{dB}$ for this typical scenario.

Next, we study the performance of our proposed method under
different compression ratios. We consider cases where the sampling
rate of the ADCs are set to 80MHz, 100MHz, and 125MHz,
respectively. Accordingly, the downsampling factor $N$ is equal to
$25$, $20$, and $16$, respectively, and the compression ratio
equals 0.32, 0.4, and 0.5, respectively. We plot the ROC curve in
Fig. \ref{fig7}, where the SNR is set to -20dB. From Fig.
\ref{fig7}, we see that the performance can be considerably
improved as the sampling rate increases.

To show the computational efficiency, we report the average run
times of the proposed method and the frequency-domain approach.
For the frequency-domain approach, the MWC consists of 8 sampling
channels. In each channel, the signal is first modulated by a
periodic PN sequence with a period of $12.5$ns (i.e. $N=25$), then
filtered using an ideal low-pass filter with a cutoff frequency of
$40$MHz, and finally sampled by an ADC with a sampling rate of
$80$MHz. We note that the MWC and the multicoset sampling scheme
have the same number of channels as well as the same sampling rate
per channel. Thus the comparison is fair. For both methods, we
collect data samples within an interval of $1$ms. We consider two
cases where the frequency resolution is set to $62.5$kHz and
$125$kHz, respectively. The experiments are conducted using MATLAB
R2015b under a laptop with 2.5GHz Intel i7 CPU and 16G RAM. We
report the average run times as well as the ROC curve in Fig.
\ref{fig8}. From Fig. \ref{fig8}, we see that our proposed method
achieves better performance than the frequency-domain approach.
Such a performance improvement is due to the fact that for our
proposed method which explicitly estimates the autocorrelation of
the original signal, a hamming window can be added to the
estimated autocorrelation sequence to enhance the power spectrum
estimation accuracy. Besides, from the reported run times, we see
that the proposed method takes about half the time needed by the
frequency-domain approach for power spectrum reconstruction. Such
an advantage in terms of the computational complexity would be
more significant for practical FPGA-based systems as our proposed
method which involves FFT operations can be more efficiently
implemented.

We also compare our proposed method with the conventional
time-domain method \cite{ArianandaLeus12}. For a fair comparison,
we assume that the multicoset sampling scheme is used by the
time-domain method to collect sub-Nyquist samples. The number of
sampling channels and the time delay parameters are the same as
described earlier. To relieve the large amount of memory required
by the conventional time-domain method, the frequency resolution
is set to 1MHz. The SNR is set to -12dB in our experiments. We
collect data samples within 0.1ms. Average run times and ROC
curves are reported in Fig. \ref{fig11}. From Fig. \ref{fig11}, we
see that our proposed method runs faster than the conventional
time-domain method. Besides, we observe that our proposed method
slightly outperforms the time-domain method, which is possibly due
to the fact that to calculate the correlation matrix, the
conventional time-domain method needs to divide the collected data
samples into a number of segments and neglects the correlation
between different segments, while our proposed method deals with
all the data samples in a batch and thus is able to obtain a
better autocorrelation estimation.

Note that our proposed method is able to reconstruct the power
spectrum without placing any sparsity constraint on the spectrum
under monitoring. To show this, we consider a multi-band signal
which consists of 32 narrowband components within the frequency
range $[0,1]$GHz. The narrowband components are generated by
passing the white Gaussian noise through band-pass FIR filters.
The bandwidth of each narrowband signal is set to 20MHz. The
center frequencies of these narrowband signals are uniformly
distributed so that no narrowband signals overlap each other. The
spectral occupancy ratio can be easily calculated as $(32\times
20\text{MHz})/(1\text{GHz})=64\%$, which indicates that the power
spectrum under monitoring is non-sparse. The sampling setup is the
same as that used in our previous examples, in which we use $M=8$
branches and the sampling rate per channel is set to $80$MHz. Fig.
\ref{fig10} shows the estimated power spectrum using noiseless
sub-Nyquist samples collected within 10ms. For a comparison, the
power spectrum estimated using noiseless Nyquist samples collected
within 1ms is also included. The NMSEs of the reconstructed power
spectrum for Fig. \ref{fig10}(a)--(b) are respectively given as
$0.0055$ and $0.0014$. From Fig. \ref{fig10}, we see that although
the spectrum under monitoring is non-sparse, our proposed method
still provides a reliable estimate of the original power spectrum.

\begin{figure*}[t]
    \centering
    \includegraphics [width=200pt]{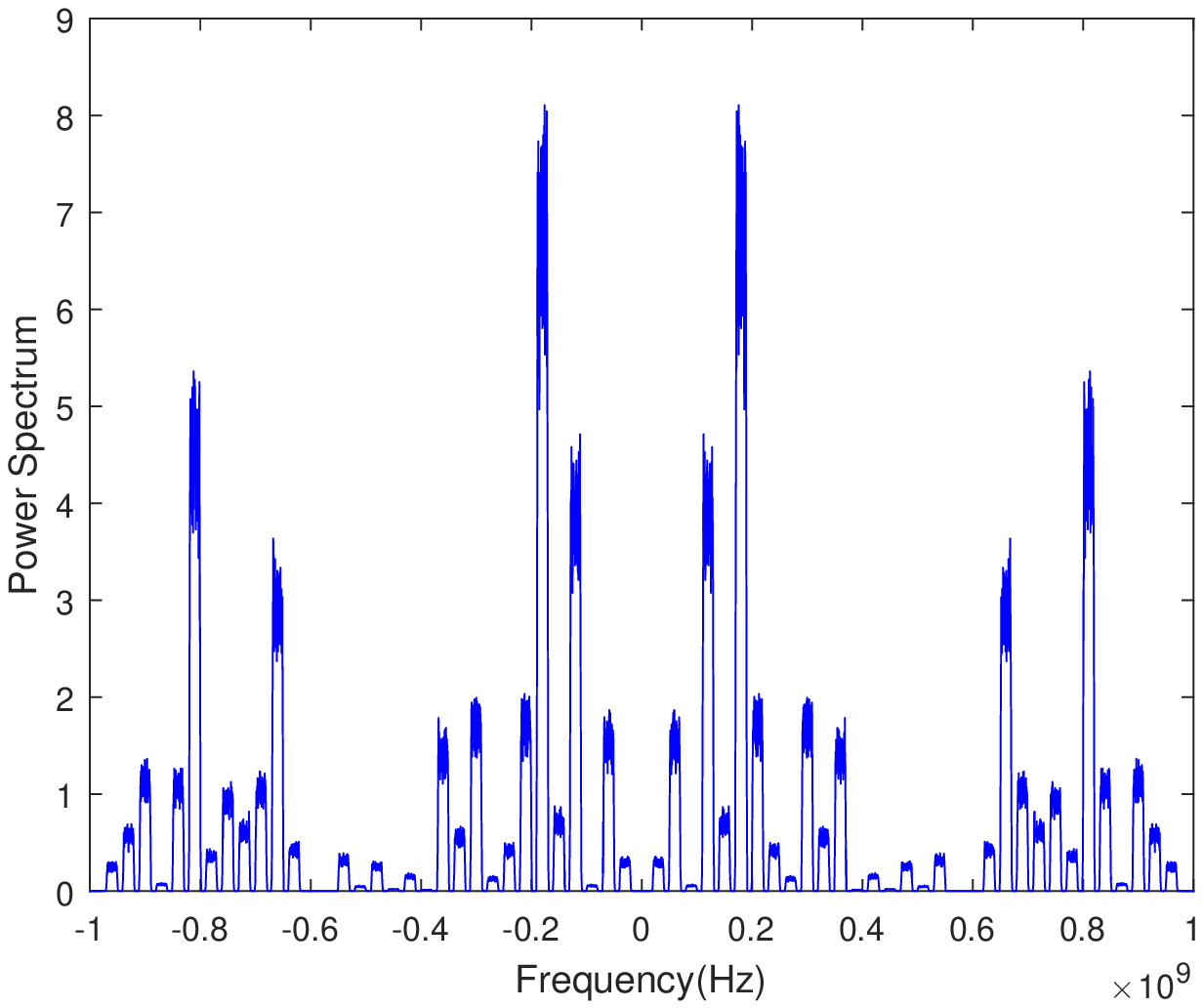}
    \hspace{15pt}
    \includegraphics [width=200pt]{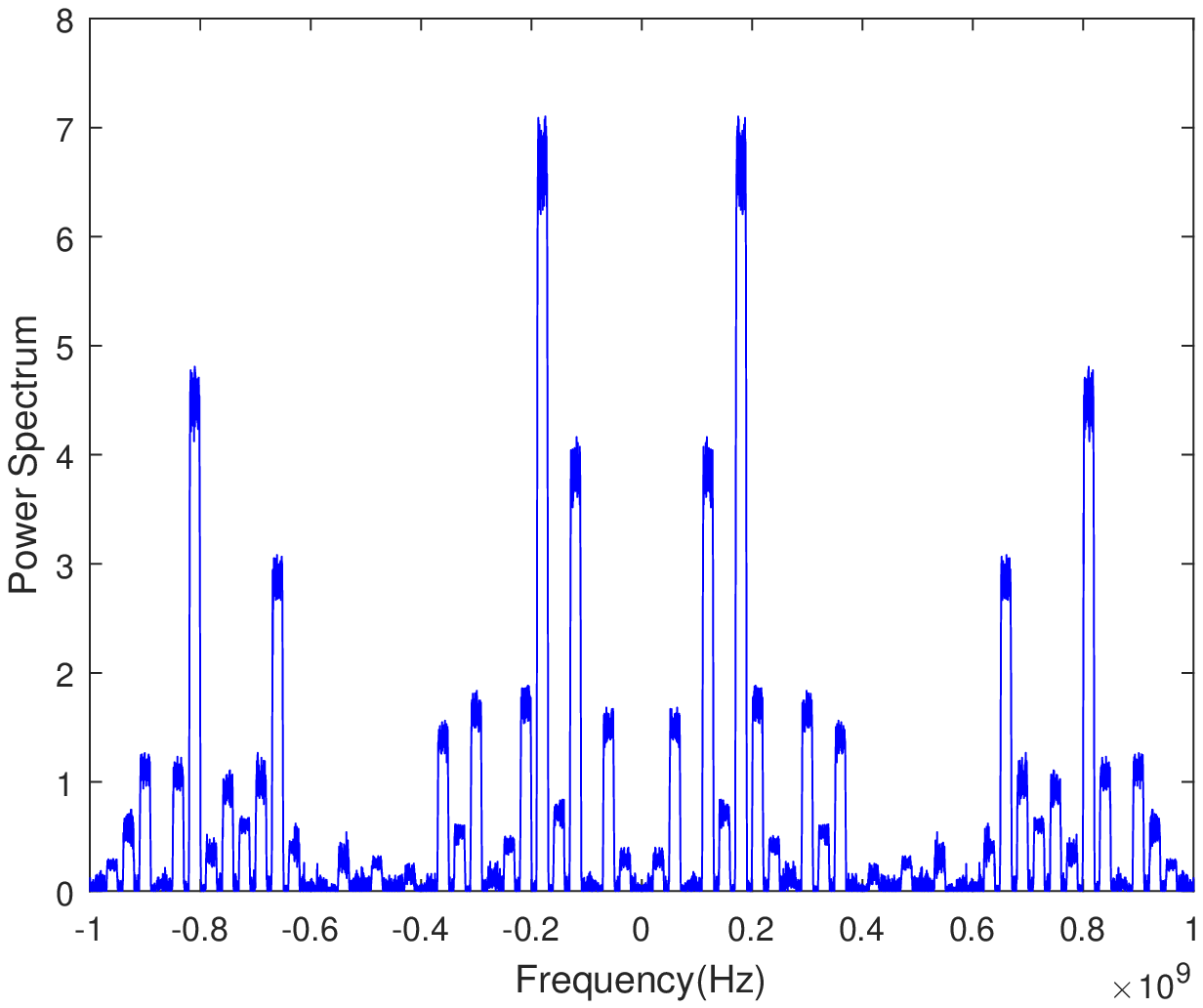}
    \caption{From left to right: Power spectrum reconstructed using noiseless Nyquist
    samples collected within 1ms; Power spectrum
        reconstructed via our proposed method with 10ms noiseless sub-Nyquist samples.}
    \label{fig10}
\end{figure*}

\subsection{Results on Cyclostationary Signals}
To show that the proposed method is applicable to cyclostationary
signals, we generate two communication signals, namely, a BPSK
signal and a QAM16 signal, which are known to be cyclostationary,
within the frequency range $[0,1]$GHz. The carrier frequencies of
these two signals are set to 130MHz and 380MHz, respectively, and
their symbol rates are set to 10M symbols per second. The SNR is
set to -5dB and the multicoset sampling architecture has the same
setup as that mentioned earlier in this section. Fig. \ref{fig9}
shows the estimated power spectra using noisy data samples
collected within 1ms and 10ms, from which we see that our proposed
method can yield an accurate power spectrum estimate of
cyclostationary signals.

\begin{figure*}[t]
    \centering
    \subfloat[]{
        \includegraphics [width=200pt]{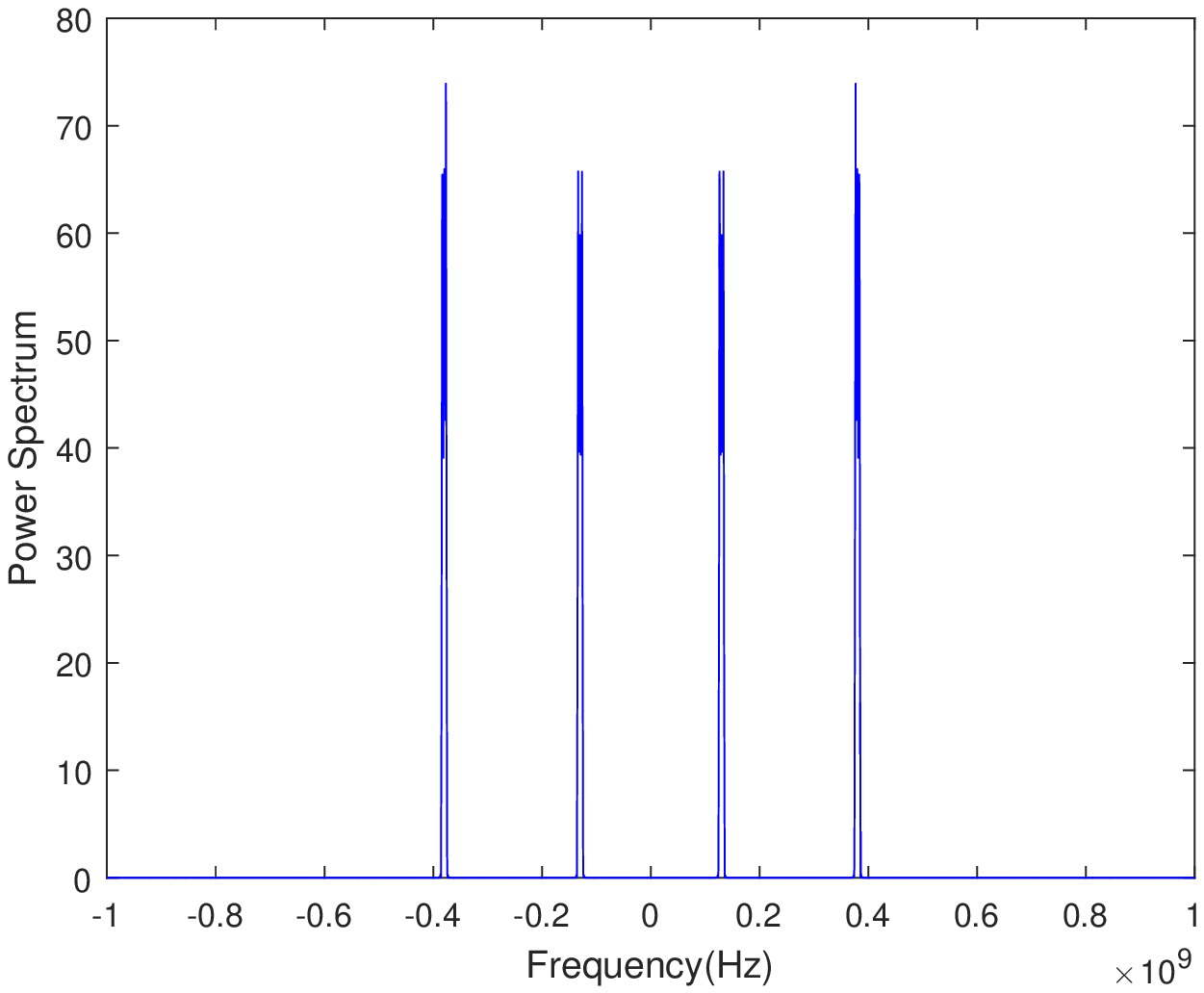}
    }
    \hspace{20pt}
    \subfloat[]{
        \includegraphics [width=200pt]{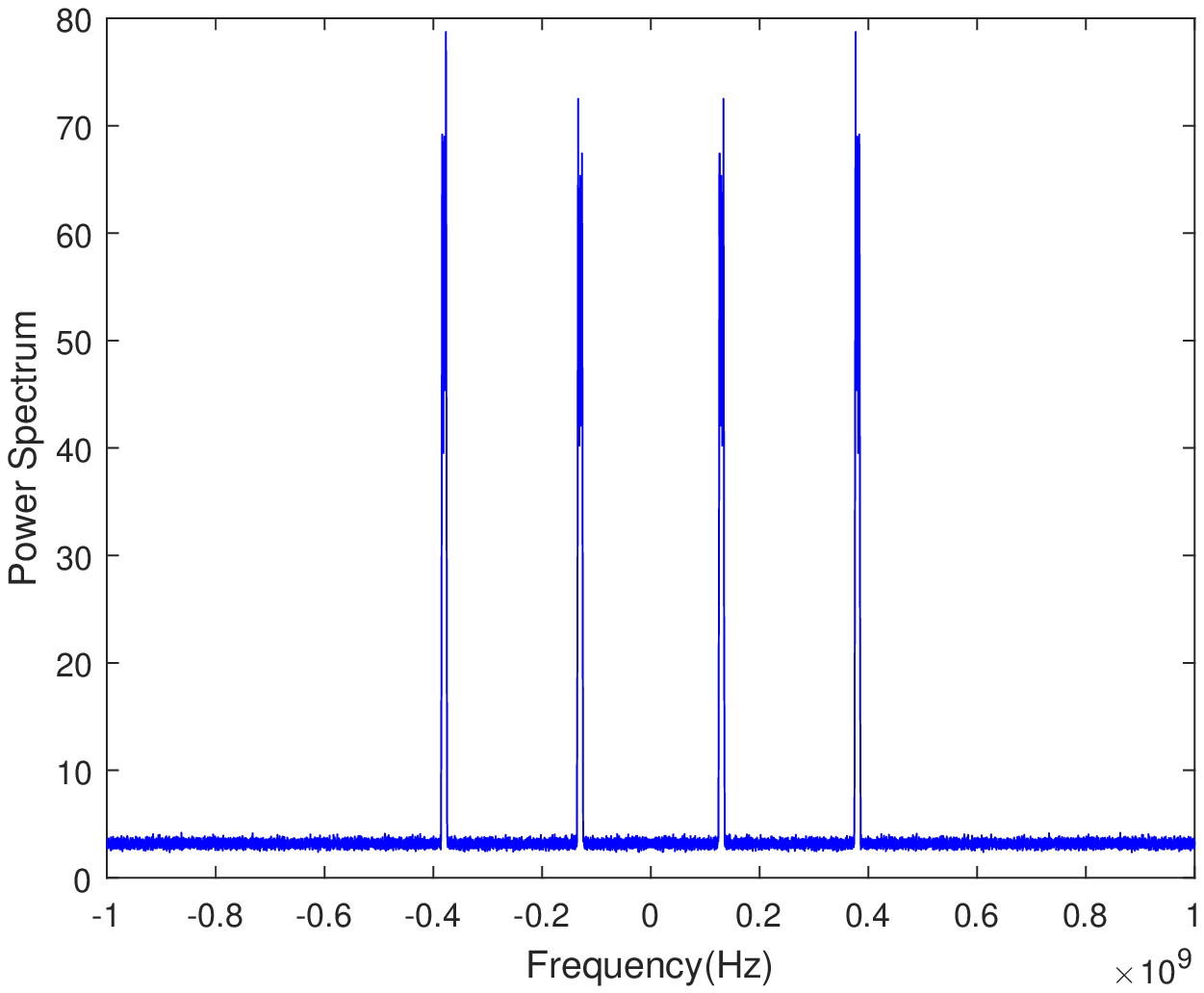}
    }
    \\
    \subfloat[]{
        \includegraphics [width=200pt]{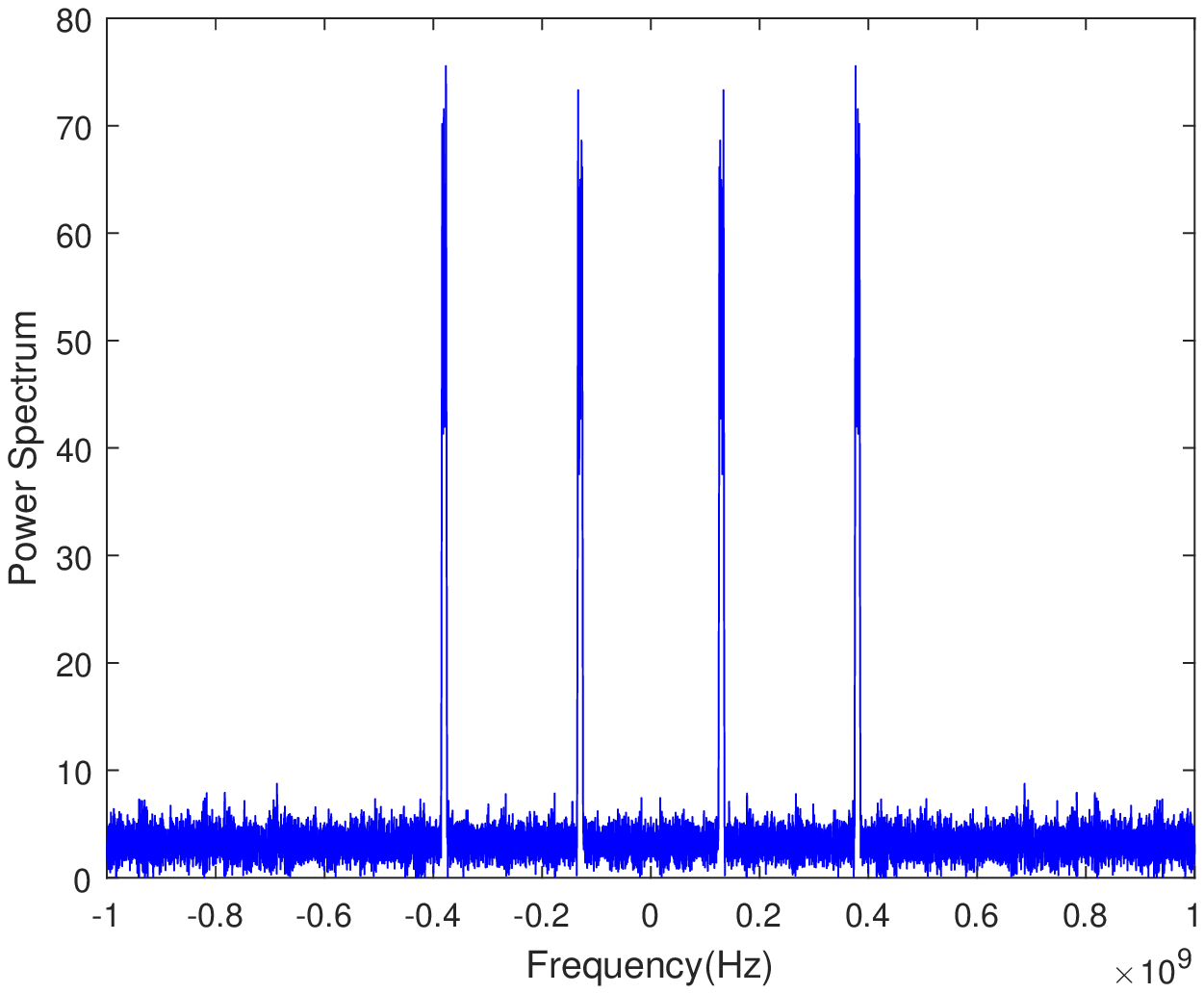}
    }
    \hspace{20pt}
    \subfloat[]{
        \includegraphics [width=200pt]{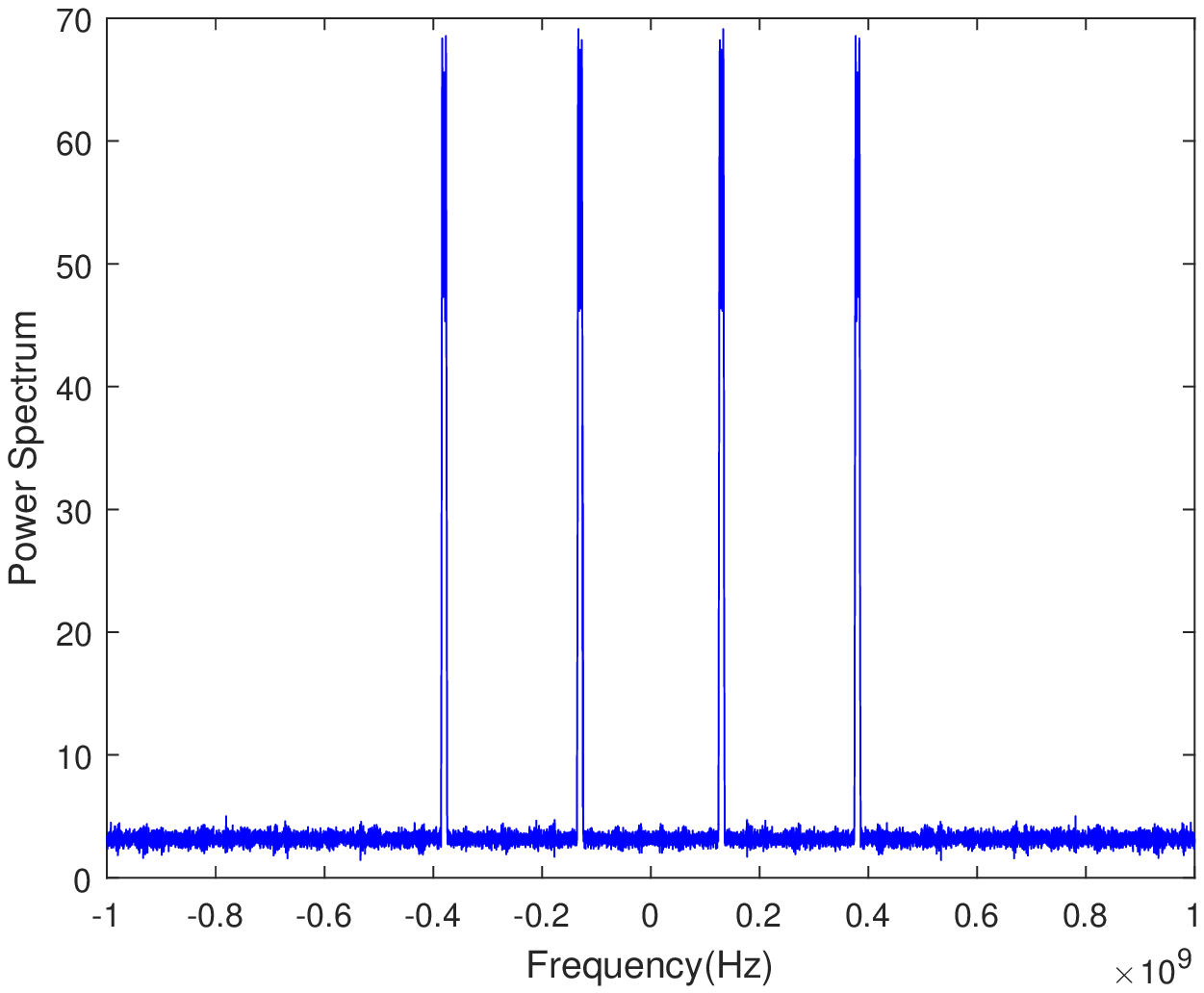}
    }
    \caption{(a) Power spectrum reconstructed using noiseless Nyquist samples collected within 1ms;
        (b) Power spectrum reconstructed using noisy Nyquist samples collected within 1ms;
        (c) Power spectrum reconstructed via our proposed method
        with 1ms noisy sub-Nyquist samples; (d) Power spectrum
        reconstructed via our proposed method with 10ms noisy sub-Nyquist samples.}
    \label{fig9}
\end{figure*}

\section{Conclusions} \label{sec:conclusions}
In this paper, we studied the problem of compressed power spectrum
estimation which aims to reconstruct the power spectrum of a
wide-sense stationary multi-band signal based on sub-Nyquist
samples. By exploiting the sampling structure of the multicoset
sampling scheme, we developed a fast compressed power spectrum
estimation method. The proposed method has a lower computational
complexity than existing methods, and can be efficiently
implemented in practical systems as its primary computational task
consists of FFT operations. Simulation results were provided to
illustrate the computational efficiency and effectiveness of our
proposed method.

\useRomanappendicesfalse
\appendices

\section{Proof of Lemma \ref{lemma1}} \label{appA}
We only need to show that for any $k\in [-LN+1,LN-1]$, there
exists at least one $n\in [0,LN-1]$ such that $I[n]I[n-k]=1$. Note
that $\{I[n]\}$ is a periodic sequence with period $N$. Hence if
$I[n]I[n-k]=1$, then for any integer $a$ such that
$n-k+aN\in[0,LN-1]$, we also have $I[n]I[n-k+aN]=1$. In other
words, if $Q_k>0$, then $Q_{k+aN}>0$. Therefore we only need to
examine $k$ in the range of $|k|\le\lfloor N/2\rfloor$. Recall
that for any $|k|\le\lfloor N/2\rfloor$, there exists $m_1,m_2$
and $c\in\{-1,0,1\}$ such that
\begin{align}
k=\Delta_{m_1}-\Delta_{m_2}+cN.
\end{align}
Set $n=\Delta_{m_1}+c'N$ and $n-k=\Delta_{m_2}+(c+c')N$, where
$c'$ is an integer such that $n\in [0,LN-1]$ and $n-k\in
[0,LN-1]$. According to the definition of $I[n]$, it is easy to
know that $I[\Delta_{m_1}+c'N]=I[\Delta_{m_2}+(c+c')N]=1$.
Consequently, we have $I[n]I[n-k]=1$ and thus $Q_k>0$.

\end{document}